\documentclass{article}
\usepackage{graphicx} 
\usepackage{amsmath}
\usepackage{tikz}
\usepackage[font={small,it}]{caption}
\usepackage{bbold, mathtools, amsthm}
\usepackage[colorlinks,hypertexnames=true,citecolor=blue]{hyperref} 
\usepackage{pgfplots}
\pgfplotsset{width=5.5cm,compat=1.9}
\usepackage[authoryear,longnamesfirst]{natbib}
\usepackage{hyperref}
\usepackage{appendix}
\setcitestyle{authoryear,open={(},close={)}}
\providecommand{\keywords}[1]{\textbf{\textit{Keywords: }} #1}
\usepgfplotslibrary{fillbetween}
\usepackage{booktabs}

\newtheorem{theorem}{Theorem}

\DeclareMathOperator{\Lagr}{\mathcal{L}}

\def\red{\color{black}}

\title{The Theory of Storage in a Power System with Stochastic Demand}
\author{Darryl Biggar and Mohammad Reza Hesamzadeh\footnote{Associate Professor, Monash University (Australia) and Professor at KTH Royal Institute of Technology (Sweden), respectively.}}

\begin{document}

\maketitle

\begin{abstract}
    \noindent Electric power systems are increasingly turning to energy storage systems to balance supply and demand. {\red But how much storage is required?} What is the optimal volume of storage in a power system {\red and on what does it depend? In addition,} what form of hedge contracts do storage facilities require? We answer these questions in the special case in which the uncertainty in the power system involves successive draws of an independent, identically-distributed random variable. We characterize the conditions for the optimal operation of, {\red and} investment in, storage and show how these conditions can be understood graphically using price-duration curves. We also characterize the optimal hedge contracts for storage units.

\end{abstract}


\section{Introduction}

With the increased penetration of Variable Renewable Energy (VRE) generation in wholesale electricity markets, there will likely be an increased role for energy storage systems. But this raises a number of key questions for power system analysts and policymakers:
\begin{itemize}
    \item How should a storage service be used? What is the efficient dispatch of a storage system? What are the private incentives for operation of storage? Do they match the overall social optimum?
    \item What is the optimal mix of storage in a power system? {\red What is the optimal amount of storage and how does it vary with, say, the level of uncertainty in the supply-demand balance?} What are the private incentives for investment in storage? Do they match the overall social optimum?
    \item What sort of hedge contracts do storage services require? What sort of hedge contracts should policymakers make available?
\end{itemize}

There is a very substantial literature exploring the implications of storage in electric power systems. But, due to the complexity of describing storage systems, and the complexity of the environment in which they operate, simple concrete and intuitive answers to these questions have yet to emerge from this literature. This article seeks to make progress by focusing on a storage system in a stylized electricity market, as described below.

In the case of conventional, controllable generation (that is, putting aside storage), there are relatively straightforward and well-known answers to these questions. For example, under certain assumptions, screening curve analysis\footnote{See, for example, the textbooks \cite{Stoft1999}, \cite{biggarhesamzadeh2014} and \cite{tanaka2024economics}} provides a simple and intuitive guide as to the optimal mix of different types of generation (baseload, mid-merit, and peaking) in an optimally-configured power system. Can we extend this screening curve analysis to include storage?

It is well known that risk-averse electricity market participants routinely seek to mitigate their risk with hedge (or forward market) contracts. The type of hedge contract required to eliminate the risk faced by generators varies with the type of generation, with baseload, peaking and wind generators all (in principle) requiring different types of hedge products. But what sort of hedge contracts do storage systems require? What hedge products might eliminate the risk of an energy storage system? This question has not been addressed in the literature.

The role of a storage system in a wholesale power market is one of \textit{inter-temporal arbitrage}.\footnote{Many storage systems also provide other ancillary services, such as frequency control. We will put these other services to one side.} That is, the storage system should charge when the spot price is low and discharge when the price is high. The future path of wholesale prices is therefore critical. In this article we will focus on the special case in which, in each dispatch interval, the realisation of the spot price (or net demand) is independent and identically distributed. In other words, a realisation of the net demand in one period tells us nothing about the realisation of net demand in the next period. This is a special case, but it is a useful starting point. {\red As we will see, even with this simplifying assumption, the classical, traditional results in power system analysis are changed in important ways. The presence of storage, it turns out, requires us to rethink some basic results in power system theory.}

This article has four main sections. Section \ref{sec:litrev} sets out a brief review of the literature. In section \ref{sec:review} we {\red start by summarizing some} basic results of conventional electricity market theory. {\red This provides a benchmark with which we compare} the corresponding results for storage services {\red later in the paper}. We will see that although {\red some principles and intuition from basic electricity market theory} remain, storage introduces new complexities which must be taken into account. Section \ref{sec:private} considers the private incentives for use of and investment in a storage service. We identify rules for the optimal private use of storage as well as rules for the optimal private investment in storage. We also identify a hedge contract which perfectly insulates the storage service from risk and discuss how that hedge contract might be provided in the market. Section \ref{sec:mix} takes the social-planner perspective and explores the question of the optimal dispatch of, and investment in, storage. This section seeks to answer the question: What is the optimal mix of storage in a power system? Section \ref{sec:examples} compares the theory set out here to the operation of large storage services in wholesale power markets in practice. Section \ref{sec:conclusion} concludes.

\section{Literature review and key contributions}\label{sec:litrev}

The engineering literature on the integration of storage into power systems is vast and is well surveyed in \cite{miletic2020operating} and \cite{sioshansi2021energy}. This literature typically uses optimization methods to find the optimal usage of different types of storage in a range of circumstances. Unfortunately, this engineering approach -- while valuable -- tends to be a `black box', yielding relatively little intuition or insight into the role of storage {\red outside the specific context considered in each paper}. 

From an economic perspective, the primary -- and most generic -- role of storage in a power system is as a form of \textit{arbitrage} -- purchasing from the market (i.e., charging) during periods of low prices and selling to the market (i.e., discharging) at times of high prices. Storage also may have a role to play in providing short-term balancing or frequency control services\footnote{The need for such balancing services arises from the fact that the typical length of the spot market dispatch interval is longer than the timescale on which the power system must be kept in balance. Reducing the length of the dispatch interval reduces the need for such balancing services. We will put aside the role of storage in providing very-short-term frequency control or balancing services.}, providing capacity in systems which directly fund capacity, in mitigating ramping constraints\footnote{E.g., \cite{shahmohammadi2018role}}, or in deferring network investment. Each of these additional roles for storage arises due to the absence of price signals, and will become less important as {\red the temporal and spatial granularity of} price signals {\red is improved} over time. This leaves the primary role for storage as a form of arbitrage.

 The simplest case to deal with is the case where the price follows a deterministic (or perfectly predictable) path. In this case, the optimal operation of the storage facility depends on whether it faces constraints on the \textit{rate} at which it can charge or discharge or constraints on the total \textit{volume} of energy that can be stored.

In the case where there are only binding constraints on the \textit{rate} (MW) of charge and discharge, in a world of perfectly predictable prices, the optimal use of a price-taking storage system is a simple binary strategy: When the spot price exceeds an upper threshold, the storage system discharges at the maximum rate; when the spot price falls below a lower threshold, the storage system charges at the maximum rate \citep{graves1999opportunities}. The difference between the two prices reflects the `round-trip' efficiency of the storage system (i.e., the energy lost in a charge-discharge cycle). The level of these price thresholds, however, depends on the overall price profile.\footnote{In the case where the storage system is perfectly efficient so that the buy threshold and sell threshold are the same, both prices are equal to the average price in the price profile.}  

In the case where there are no rate limits (and only binding constraints on the total \textit{volume} of energy that can be stored) the optimal use of the storage system is quite different. In this case the optimal strategy is to charge to the maximum volume at any \textit{local minimum} in the price profile and to discharge to zero at any \textit{local maximum} in the price profile. Of course, this requires knowledge of the future prices to identify local maxima and minima. Nevertheless, a number of articles use this `perfect foresight' assumption to estimate the arbitrage value of energy storage \citep{mcconnell2015estimating, das2015assessing, sioshansi2009estimating, walawalkar2007economics, figueiredo2006economics, hu2010optimal, lamp2022large}.

In practice, future spot prices are not perfectly predictable. If we allow for uncertainty in future outcomes, stochastic dynamic programming techniques must be used (e.g., \citep{van2013optimal, sioshansi2013dynamic, xi2014stochastic, mokrian2006stochastic, lohndorf2010optimal}). However, these quickly become complex and difficult to solve, necessitating quite strong simplifying assumptions. In the analysis that follows we make a key simplifying assumption about the nature of future spot prices. Specifically, we assume that all draws of the spot price are independent and identically distributed.

Most of the literature on storage assumes price-taking behaviour. An important strand of the literature on storage relaxes this assumption to consider strategic behaviour by storage \citep[See][]{miletic2020operating,andres2023storing, siddiqui2019merchant, karaduman2020economics,balakinroger2025}. Finally, we note that another strand of the literature focuses on identifying the optimal level of investment in storage, sometimes known as the Storage Expansion Plan. See for example, the surveys by \cite{haas2017challenges} and \cite{sheibani2018energy}.

Despite the size of the literature on energy storage, we consider that relatively few {\red general} insights or intuition have emerged. In part, this is due to the fact that storage systems differ widely in both their maximum rate of production and their maximum volume of stored energy (as well as other factors such as the round-trip efficiency, rate of deterioration etc.). In addition, storage systems operate in a complex, uncertain and changing environment. In this article we aim to develop concrete insights and intuition about the optimal use of, and investment in, storage in certain special cases. Specifically, drawing inspiration from screening curves, we focus on the case of a storage system that is not rate-limited and a power system in which there is no time-evolution of prices.

Risk-averse wholesale electricity market participants routinely seek to be hedged -- that is, to offset their cashflow risks through financial hedge products. Whereas a range of financial hedge products are commonly available to different types of generation, to our knowledge the question of hedges for energy storage products has not been addressed in the literature. We show that in the special case considered in this article the optimal operation of the storage system can be characterised. We draw on this result to identify a set of hedge contracts which perfectly hedge the risks of storage systems.

Finally, we characterise the optimal level of investment in storage capacity (MWh) in this context, and show how this capacity relates to familiar ideas of screening curves.

{\red The key assumption at the base of the analysis that follows is that there is only one source of uncertainty in the power system (here, represented by variation in load) and that uncertainty is independently and identically distributed across intervals. This is a strong assumption. It implies that the realisation of load in one period provides no information about the realisation of load in the next. In practice, load is not iid over time. Rather, load follows regular daily, weekly and seasonal cycles, with a degree of uncertainty around those underlying cycles due to factors such as the weather, and planned and unplanned outages.

We justify this key assumption on the basis that: (a) we seek to develop understanding and intuition in the operation of, and investment in, storage and this is easiest to achieve in clear special cases; (b) this assumption is consistent with common assumptions in power system theory (i.e., the theory underlying screening curves and many other familiar results); and (c) the operation of storage rapidly becomes significantly more complex when the load follows a more complex stochastic process. The results set out here represent an important benchmark against which we can compare storage operation rules under more complex stochastic processes. At the same time, the results set out here represent a stepping stone showing how familiar results in power system theory (e.g., screening curves) extend to the case of storage under strong assumptions.

We make two further simplifying assumptions -- the assumptions (a) that the storage faces no efficiency losses and no degradation, and (b) that the storage is not rate limited (that is, it can charge or discharge as much as desired in the course of a single dispatch interval). These assumptions are made for convenience. As we will see, there are important lessons even in the simplified case studied here. The implications of these assumptions can be explored in future research. }

The key contributions of this article are as follows:
\begin{itemize}
    \item We characterize the efficient operation of a non-rate limited storage system, in a simplified power system and provide a simple and intuitive characterisation of the private incentive for investment in storage.
    \item We identify a financial hedge contract which can provide a perfect hedge for storage and show how this hedge can be constructed out of a combination of caps and floors.
    \item We characterize the socially-efficient level of investment in storage and provide a simple graphical interpretation.
\end{itemize}

\section{Review of basic electricity market theory}\label{sec:review}

This section summarises {\red well-known, fundamental results in the theory of electricity markets. The purpose is to provide a benchmark or point of departure for the key results on storage that follow. This section} introduces the notation and the flow of logic in a familiar setting before {\red later sections} extend the results to the case of storage. Table \ref{tab:nomenclature} summarises the notation used in this article.

\subsection{Social incentives for usage of, investment in, and hedging of generation}\label{sec:privgenonly}


Let's suppose that we have a range of generation technologies, labelled $i=1, \ldots, N$. Each generation technology $i$ has a constant variable cost of production $c_i$ (\$/MWh) and a constant fixed cost of adding capacity $f_i$ (\$/MWh).

As is conventional, uncertainty in the power system is introduced in the form of an uncertain, inelastic load. The realisation of the load each interval is given by the random variable $L$. This random variable is assumed to be independent and identically distributed each interval. The demand curve is equal to some large, fixed value $V^{LL}$ (the value of lost load), up to the rate of consumption $L$, at which point the demand drops to zero. We can model this as a generator with zero fixed cost and a variable cost equal to $V^{LL}$.

We can define the \textit{state} of the power system to be all the factors that affect supply and demand (and therefore affect the price). In the model we are using here, the only source of uncertainty (and therefore the only variable defining the state of the power system) is the realisation of the load $L$.

In conventional power system modelling, the state of the power system in each time interval is (implicitly) assumed to be independent and identically distributed from one time interval to the next. As a consequence, the realisation of the price in each interval is independent and identically distributed -- in other words, each realisation of the price is a drawing from the same distribution. (As we will see shortly, this assumption will not hold in the presence of storage).


Let's suppose that the capacity of generation type $i$ is given by $K_i$ (MW). Given this mix of generation, each interval the task of the social planner is to observe the load $L$ {\red (MW)} and to choose the dispatch of the available generation to minimise the total cost. Let $DC(L|\vec{K})$ {\red (\$/h)} denote the total cost of dispatch\footnote{Strictly: The rate at which expenditure is incurred.} when the power system is operated efficiently, the total load is $L$ and the mix of capacities is given by $\vec{K}=\{K_1, K_2, \ldots, K_N \}$. The efficient dispatch of each generation type $Q_{iL}$ (MW) is the solution to:
\begin{align}
    \forall L,\vec{K}, \; DC(L|\vec{K}) &\equiv \min_{Q_{iL}} \sum_i c_i Q_{iL}  \text{ subject to: }\nonumber\\
    &\sum_i Q_{iL}=L\nonumber\\
    \forall i,\; &0 \le Q_{iL} \le K_i\label{eqn:dclk}
\end{align}
The solution to this problem is the familiar merit-order dispatch: Generators are ranked in order from the cheapest (lowest variable cost $c_i$) to the most expensive (highest variable cost $c_i$), with the lower-variable costs generators used to capacity before the higher variable cost generators are used at all. For each realisation of the load $L$, there is a corresponding `system marginal cost' (equal to the variable cost of the most expensive generator that is dispatched), which can be identified with the wholesale spot price $P$ {\red (\$/MWh)}. Specifically, we can define the `price' to be the change in dispatch cost for a small change in the load:
\begin{equation}
    \forall L,\vec{K}, \; P(L|\vec{K}) \equiv \frac{\partial DC}{\partial L}(L|\vec{K})\label{eqn:plk}
\end{equation}
Then it follows that the optimal dispatch of each generation technology $i$, $Q^*_{iL}$ satisfies:
\begin{equation}
    \forall i, L, \; Q^*_{iL} = \begin{cases}
        0, &\text{if } P(L|\vec{K})<c_i\\
        K_i, & \text{if } P(L|\vec{K})>c_i\\
        \text{in the range }[0,K_i], &\text{if } P(L|\vec{K})=c_i
    \end{cases}\label{eqn:sod}
\end{equation}

We will assume that the capacity of each generator must be chosen before the load is realised. The optimal mix of each generation type is the choice of $\vec{K}$ which minimises the total cost of the power system {\red (denoted $V(\vec{K})$)} which we can write as:
\begin{equation}
    \forall \vec{K}, \; V(\vec{K}) \equiv E[DC(L|\vec{K})]+\sum_{i} f_i K_i\label{eqn:vf1}
\end{equation}
Here the expectation is over all possible realisations of the load $L$.

Solving this problem (minimising over $K_i$ for each $i$), we find the familiar condition that, in the optimal mix of generation, for each generation type $i$, the area under the price-duration curve and above the variable cost of production $c_i$ is equal to the fixed cost of capacity $f_i$. In other words, for each technology $i$, the optimal capacity $K_i$ of that technology is where the area under the price-duration curve is equal to the fixed cost of that technology $f_i$:
\begin{align}
    \forall i, \; \frac{\partial V}{\partial K_i} &= E[\frac{\partial DC}{\partial K_i} (L|\vec{K})]+ f_i\nonumber\\
    &=  -E[ (P(L)-c_i) I(P(L)\ge c_i)] + f_i = 0\label{eqn:soi}
\end{align}
Here $I(\cdot)$ is the indicator function which takes the value one when the expression in brackets is true and zero otherwise. {\red $E[ (P(L)-c_i) I(P(L)\ge c_i)]$ is the area under the price-duration curve and above the variable cost of the generator.}

This leads, in a straightforward manner, to the concept of `screening curves' and familiar results such as (a) in an optimally-configured power system, the shape of the price-duration curve is determined entirely by cost structure of the generation types (and not by load characteristics); and (b) the average price (the area under the price-duration curve) is entirely determined by the fixed and variable cost of the cheapest generation.\footnote{See, for example, the textbook treatment in \cite{biggarhesamzadeh2014}.}

\subsection{Private incentives for usage of, and investment in, generation}\label{sec:social}

Now let's consider the \textit{private} incentives to operate, or to invest in, generation.

Let's suppose we have a private, for-profit, controllable generator with a maximum rate of production $K$ (MW) (said to be the `capacity' of the generator).  {\red The variable cost of the generator is $c$ (\$/MWh) and the fixed cost per unit of capacity is $f$ (\$/MWh).} Each period the generator observes the spot price $P$ (\$/MWh) and makes a decision as to how much to produce. If the generator produces at rate $Q$, it receives a cashflow stream at the rate $\pi$ (\$/h) where $\pi$ is given by:\footnote{It follows that each dispatch interval the generator receives revenue $\pi \Delta$ where $\Delta$ is the length of the dispatch interval, discussed further below.}
\begin{equation}
    \pi(Q|P) = PQ - cQ\label{eqn:1}
\end{equation}


Let $V(P|K)$ be the (one-period) cashflow stream of the generator (before fixed costs) when it is operated efficiently (i.e., profit maximising) and the spot price is $P$. Given the spot price $P$, and capacity $K$, the generator chooses a profit-maximising rate of production $Q_P$:
\begin{equation}
    \forall P,K, \; V(P|K) \equiv \max_{Q_P} \pi(Q_P|P) \text{ subject to: } 0 \le Q_P \le K\label{eqn:osc}
\end{equation}
This has the familiar solution:
\begin{equation}
    Q^*_P = \begin{cases}
        0, & \text{if }P<c\\
        K, & \text{if }P > c\\
        \text{in the range }[0,K], &\text{if } P=c\\
    \end{cases}\label{eqn:pod}
\end{equation}

When it comes to investment, we will make the conventional assumption that the capacity of the generator $K$ must be chosen before the spot price is realised. The total expected payoff to the generator with capacity $K$ (after including fixed costs) is $E[V(P|K)]-fK$ where the expectation is taken over all possible values of the spot price $P$ (recall that this distribution is stationary).

The generator owner has an incentive to add capacity (or new generation capacity has an incentive to enter the market) up to the point where $E[\frac{\partial V}{\partial K}(P|K)]=f$. Drawing on equations \ref{eqn:1}, \ref{eqn:osc} and \ref{eqn:pod}, this is equivalent to the condition that the area under the price-duration curve and above the variable cost of the generator $c$ is equal to the fixed cost of the generator $f$:
\begin{equation}
    E[ (P-c)I(P \ge c)] = f\label{eqn:poi}
\end{equation}

It follows that, in the free-entry equilibrium, for every generation technology, the area under the price-duration curve and above the variable cost of that technology must be equal to the fixed cost of that technology.


We can observe that the socially-optimal dispatch and investment decisions (equations \ref{eqn:sod} and \ref{eqn:soi}) are identical for the private profit-maximising dispatch and investment decisions (equations \ref{eqn:pod} and \ref{eqn:poi}), so a competitive market of profit-maximising agents will achieve the socially-optimal outcome.

\subsection{Perfect hedges for generators}

Let's turn now to the question of an `optimal' hedge contract for this generator. As the spot price is stochastic, the generator above experiences a variable cashflow given by $\pi(Q^*_P|P)$. If this generator is risk averse it would like to reduce or eliminate this risk. We will seek a `perfect' hedge contract that completely eliminates the risk faced by the generator.\footnote{\red In practice, depending on the supply and demand for insurance products, there is likely to be some premium associated with a perfect insurance product. As a result, market participants, even if risk averse may not choose to be perfectly hedged. Nevertheless, the case of a perfect hedge provides a useful benchmark and guide as to how to design hedging instruments. In effect, we are assuming that there exist risk neutral entities (which might be financial institutions) who are willing to take on risk with little need for compensation and who therefore can meet the hedging needs of other market participants at relatively low cost.\label{fn:hedge}}

A hedge contract is a financial payment from the generator to another party, which depends on the spot price, and potentially on the dispatch of the generator, $H(Q, P)$. {\red If the spot price is $P$ and the generator chooses to produce at rate $Q$, the hedged cashflow of the generator is equal to:
\begin{equation}
    \pi(Q|P)-H(Q,P)
\end{equation}}

A perfect hedge contract is a hedge contract $H(Q,P)$ which has two key characteristics:
\begin{enumerate}
    \item First, the hedge contract must eliminate the variability in the spot price when the generator is operating efficiently:
        \begin{equation}
            \pi(Q^*_P|P)-H(Q^*_P,P)=\text{constant}\label{eqn:perfhc}
        \end{equation}
    \item Second, the hedge contract must not affect or distort the efficient production decision of the generator. In other words,the optimal supply curve $Q^*_P$ which is a solution to equation \ref{eqn:osc} (as given in equation \ref{eqn:pod}) must also be a solution to the task of maximising the hedged profit of the generator:
        \begin{equation}
            V^H(P|K) \equiv \max_{Q_P} \pi(Q_P|P)-H(Q_P,P) \text{ subject to: } \forall 0 \le Q_P \le K\label{eqn:osch}
        \end{equation}
\end{enumerate}
Because, by definition, the optimal rate of production $Q^*_P$ maximises the cashflow $\pi(Q|P)$ (subject to the constraints), it follows that the second condition can be satisfied simply by making the hedge contract independent of the actual or out-turn production of the generator. In other words, a sufficient condition for the second criterion is as follows:
\begin{equation}
    \forall P, \; \frac{\partial H}{\partial Q}(Q,P)=0
\end{equation}
This is an important condition which is satisfied by most conventionally-trade hedge contracts (such as swaps, caps, and floors).\footnote{This condition is not satisfied by some common hedge contracts for wind generators, and this has been heavily criticised in by several commentators including \cite{newbery2022}.}

In the case above, where the generator has a constant variable cost, the cashflow of the generator is equal to:
\begin{equation}
    \pi(Q^*_P|P)=(P-c)Q^*_P = \begin{cases}
        (P-c)K, & \text{if } P>c\\
        0, &\text{otherwise}
    \end{cases}
\end{equation}
This cashflow can be perfectly hedged by a Cap contract with a strike price $c$ and a volume $K$. A Cap contract  with a strike price $c$ and a volume $K$ has the following payoff.
\begin{equation}
    Cap(P|c,K) \equiv (P-c) K I(P\ge c)
\end{equation}
A cap contract with a strike price $c$ and a volume $K$ is a perfect hedge for this generator:
\begin{equation}
    \pi(Q^*_P|P)=Cap(P|c,K)
\end{equation}

More generally, for any generator with an upward sloping convex cost function $c(Q)$ \cite{biggar2022integrated} show that it is possible to come arbitrarily close to a perfect hedge using a portfolio of cap contracts {\red with strike prices that are arbitrarily close together}. Similarly, \cite{biggarhesamzadeh2014} show how loads can be perfectly hedged using a portfolio of floor contracts.

\begin{table}[h!]
    \centering
    \caption{Nomenclature used in this article}
    \begin{tabular}{p{2.5 cm} p{2 cm} p{5 cm}}
      \toprule
      Notation  & Unit & Explanation \\
      \midrule
      $\Delta$ & hours & length of dispatch interval \\
      $c$,$c_i$ & \$/MWh & variable cost of generating plant\\
      $f_i$,$f_S$ & \$/MWh & fixed cost of generation technology $i$, storage\\
      $L$ & MW & Random variable representing uncertain load\\
      $L^{max}, L^{min}$ & MW & Upper and lower bound on the load\\
      $S$ & MWh & State of charge of the system-wide storage assets (which evolves according to a Markov process)\\
      $S^+(S,L), S^+_{SL}$ & MWh & New state of charge given old state of charge and the realisation of load\\
      $S^{max}, S^{min}$ & MWh & Upper and lower bound on the state of charge\\
      $M(S,T), M_{ST}$ & Prob & Matrix describing the probability of transition from state of charge $S$ to state of charge $T$.\\
      $P(L), P_L$ & \$/MWh & spot price in current dispatch interval (which, in the simplest case is a function of the load $L$)\\
      $P(S,L), P_{SL}$ & \$/MWh & Spot price when the overall state of storage is $S$ and the load is $L$\\
      $Q$, $Q_i$ & MW & rate of production of generating plant\\
      $K$ & MW & maximum rate of production (capacity) of generating plant\\
      $x(s)$ & Prob & Probability the system is in state $s$ in the stationary distribution\\
      $\delta$ & & Rate of time discount per interval\\
      $DC(L|\vec{K})$ & \$/h & Total dispatch cost when meeting total load of $L$ in a context in which the generation capacities are $\vec{K}$\\
      \bottomrule
    \end{tabular}
    \label{tab:nomenclature}
\end{table}

\section{Optimal social use of and investment in storage in a power system}\label{sec:mix}

Let's now extend this basic power system theory to incorporate storage. The introduction of storage in the power system has one important consequence for how we model the power system -- we can no longer model the power system as being in the same state each period. Now, the level of storage in the power system affects the balance of supply and demand and therefore affects outcomes.

Extending the analysis above, we will consider a power system in which the state of the power system is determined entirely by the state of charge of the storage $S$. 

To keep things simple, We will consider a stylized storage system which is not limited in the \textit{rate} (MW) at which it can inject or withdraw power, but rather is limited in the total \textit{volume} of energy stored (MWh). This storage system will be assumed to be perfectly efficient (the energy withdrawn from the power system is equal to the energy that is able to be re-injected). There are constant returns to scale -- there is a fixed cost per unit of storage capacity which we express as the cost $f_S$ \$/MWh per hour. The storage capacity of the power system is defined as the difference between the maximum $S^{max}$ and minimum $S^{min}$ state of charge.

Let's suppose that the state of charge of the power system as a whole at a given point in time is given by $S$ {\red (MWh)}. At this point in time all of the other uncertainty in the power system (e.g., due to the variation in load) is realised. Let's suppose that this uncertainty is represented in the random variable $L$. Given this opening state of charge $S$, and realisation of uncertainty $L$, the closing state of charge is assumed to be given by $S^+_{SL}$ {\red (MWh)}.

{\red It follows that: if the load is $L$, the current state of charge is $S$, and the storage facility is following the rule $S^+_{SL}$, the net demand (that is the demand including the contribution from storage) is:
\begin{equation}
    L-\frac{S-S^+_{SL}}{\Delta}
\end{equation}
Here $\Delta$ (h) is the length of the dispatch interval.}

The transition from the opening state of charge $S$ to the closing state of charge $S^+$ can be represented as a Markov process. The probability of transitioning from an opening state of charge $S=s$, to a closing state of charge $S^+=t$ is given by $M_{st}$, defined as follows:
\begin{equation}
    M_{st} \equiv E_L[ I(S^+_{sL}=t) ]\label{eqn:mst}
\end{equation}

From any starting state of charge, after the power system evolves through any finite number of intervals the power system can be described as being in a probability distribution over possible future states, which we label $x(s)$. A key concept in the analysis that follows will be the concept of the \textit{stationary distribution}. A probability distribution over states is stationary if it satisfies the following condition:
\begin{equation}
    \forall t, \;\; \int M_{st} x(s) ds = x(t)\label{eqn:s}
\end{equation}

As before, we will start by adopting the perspective of the system planner. The task of the system planner is to find the socially-efficient use of, and investment in, storage in the context of an overall power system. 

\subsection{Optimal use of storage in a power system}

As before, let's assume that the dispatch cost in a power system when consuming at rate $L$ is $DC(L|\vec{K})$. 
Recall that the network `price' (also known as `System Marginal Cost') is the slope of the dispatch cost with respect to the load (equation \ref{eqn:plk}).

%

Let's suppose that the total state of charge of the storage is $S$. The task of the system operator is to find the optimal storage supply curve $S^+_{SL}$ which solves the following Bellman equation:
\begin{align}
    \forall S,\vec{K},K_S,\; V(S|\vec{K},K_S) &\equiv \min_{S^+_{SL}} E[ DC(L-(S-S^+_{SL})/\Delta|\vec{K}) + \delta V(S^+_{SL}|\vec{K},K_S)]\nonumber\\
    \text{subject to: }& \forall L, \;S^{min} \le S^+_{SL} \le S^{min}+K_S\equiv S^{max}\label{eqn:bm2}
\end{align}
{\red Here $\delta$ is the time discount rate. The expectation $E[\cdot]$ is taken over all possible realisations of the load $L$.}

It turns out that the socially-optimal dispatch of the storage system is to charge the storage system to its maximum at times when the realisation of demand is low, so that the resulting spot price (even when the storage is charging at the maximum rate) is below the average price in the next period (discounted for the time value of money). Similarly, it is optimal to discharge the storage system to its minimum level at times when the realisation of demand is high so that the spot price (even with discharging) is above the average price in the next period discounted for the time value of money.

Specifically, let's suppose that, given the state of charge $S$ and the realisation of the load $L$, the strategy of the system operator is to choose the closing state of charge $S^+_{SL}$. We can define $P_{SL}$ be the spot price in the current period when the state of charge is $S$, the realisation of the load is $L$, and the system operator follows the strategy $S^+_{SL}$:
\begin{equation}
 \forall S,L, \; P_{SL} \equiv \hat{P} \left(L-(S-S^+_{SL})/\Delta\right)   \label{eqn:psl}
\end{equation}
Here $\hat{P}(L)=\frac{\partial DC}{\partial L}$ {\red (from equation \ref{eqn:plk})} is the spot price when the storage is neither charging nor discharging. 

The price in the next period $P_{S^+_{SL}L'}$ depends both on the state of charge at the  start of the next period ($S^+_{SL}$) and on the realisation of the load in the next period $L'$. The expected price in the next period will be denoted as follows:
\begin{equation}
    \forall S,L, \; E[P^+_{SL}] \equiv E_{L'}[P_{S^+_{SL}L'}]\label{eqn:eps}
\end{equation}

As shown in theorem \ref{thm:1}, {\red for any given state of charge $S$ and realisation of load $L$),} the optimal strategy $S^+_{SL}$ depends on the relationship between the current price {\red when the load charges or discharges to the level $S^+$} and the expected price in the next period {\red (which also depends on the closing state of charge $S^+$),} discounted by the time value of money. Specifically, the optimal strategy must satisfy the following condition:
    \begin{align}
        \forall S, L, \; S^+_{SL}=\begin{cases}
            S^{max}, & \text{if } P_{SL}<\delta E[P^+_{SL}]\\
            S^{min},  & \text{if } P_{SL}>\delta E[P^+_{SL}]\\
            \text{in the range }[S^{min},S^{max}],  & \text{if } P_{SL}=\delta E[P^+_{SL}]
        \end{cases}\label{eqn:optsplus}
    \end{align}
{\red This is the storage analogy of equation \ref{eqn:pod}.}

Note that this is a recursive condition. The optimal strategy $S^+$ affects (a) the price realisation today (equation \ref{eqn:psl}) and (b) the expected price realisation tomorrow (equation \ref{eqn:eps}). This is a complex relationship. A high rate of charge in the current period will tend to raise the price, and raise the closing state of charge, which tends to lower the expected price in the next period, and vice versa.

We can, in principle, compute the optimal strategy $S^+_{SL}$ through an iterative process. Specifically, given an initial guess at $S^+_{SL}$, for each possible value of {\red $S$ and $L$}  we can compute the {\red current spot price $P_{SL}$ from equation \ref{eqn:psl} and the expected price in the next period $E[P^+_{SL}]$ from equation \ref{eqn:eps}. The new optimal strategy ${S^+}'_{SL}$ is given by equation \ref{eqn:optsplus}, which can be written in a slightly more compact form as follows:}
\begin{equation}
    \forall S, L,\; {S^+}'_{SL}=\left| S + \Delta(\hat{P}^{-1}(\delta E[P^+_{SL}]) - L) \right|_{S_{min}}^{S^{max}}
\end{equation}
Here $|\cdot|_L^H$ means the expression in brackets is bounded by the bounds $L$ and $H$. Where this process converges it converges to the optimal strategy, given by equations \ref{eqn:psl}-\ref{eqn:optsplus}.

Figure \ref{fig:3} illustrates the determination of the optimal strategy in the case where the storage capacity is small relative to the variation in load.  As we noted above the optimal closing state of charge depends on the relationship between the current spot price and the next-period expected spot price. In figure \ref{fig:3} the blue line represents the price-duration curve and the red represents the expected future spot price discounted by the time value of money (again, reflected as a price-duration curve).\footnote{{\red The horizontal axis represents the load duration. For any given load $l$, the load duration is in a one-to-one relationship with the actual load and is given by $Pr(L\ge l)$.}}

The right hand side of figure \ref{fig:3} illustrates the case where  the realisation of load $L$ is low (so the current spot price is low) -- lower, even, than the expected future spot price (taking into account that the storage will be full in the next period). In this case the optimal strategy is to charge the storage to the maximum $S^+_{SL}=S^{max}$. Similarly, when the realisation of the load is high (on the left in figure \ref{fig:3}) so that the current spot price (even taking into account the discharge of the storage) is high -- above the expected future spot price -- the optimal strategy is to discharge the storage to the minimum $S^+_{SL}=S^{min}$. The gray areas of the graph in figure \ref{fig:3} represent these extreme cases -- where the storage ends up fully charged or discharged. It is only in these regions that the storage expects to receive a non-zero cash-flow.

For intermediate values of the realisation of load, the storage charges or discharges to the point where the spot price today (in the current interval) is \textit{equal} to the expected future spot price (discounted by the time value of money) tomorrow (i.e., in the next interval). At these intermediate levels of load the storage expects to receive, on average, no revenue or incur no costs.

\begin{figure}[!ht]
    \begin{tikzpicture}
\begin{axis}[width=10 cm, height=7cm,
    xlabel={Duration},
    ylabel={Price},
    ymin=0, ymax=200,
    legend pos=north east,
    ymajorgrids=true,
    grid style=dashed,
]

\addplot[name path = A,
    color=blue]
    table[x=Load,y=PSmid,col sep=comma] {data/CaseB-Smax10.csv};
\addplot[name path= B,
    color=red]
    table[x=Load,y=EPSmid,col sep=comma] {data/CaseB-Smax10.csv};    
    \legend{{$P$},{$\delta E[P]$}}
\addplot[gray!20] fill between[of=A and B];
    
\end{axis}
\node[right, text width = 1.8 cm] (A) at (4.0,1) {Storage charges to $S^{max}$};
\node [below, text width = 2.5 cm] (B) at (3.5,5) {Storage discharges to $S^{min}$};
\draw [->] (A)--(6.7,2);
\draw [->] (B)--(1.5,3.5);

\end{tikzpicture}
    \caption{Illustration of the determination of $S^+_{SL}$ in the case where the storage capacity is 10\% of the variation of the load). Load is uniformly distributed on [0,100]. Raw price is assumed to be a linear function of load {\red $P(L)=20+1.5L$. Initial storage level in the middle of the range.}}
    \label{fig:3}
\end{figure}
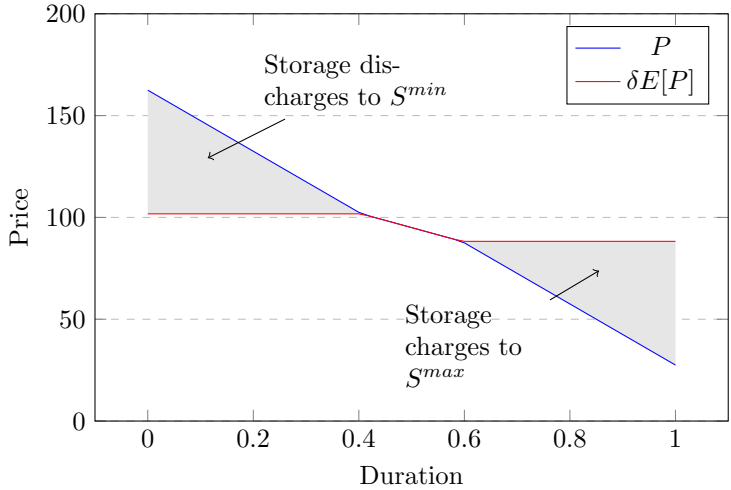

Figure \ref{fig:3} illustrates the case where the starting state of charge $S$ is in the middle of its range. But the optimal strategy {\red ($S^+_{SL}$ and the corresponding price $P_{SL}$)} 
 depends on the opening state of charge $S$ {\red (as well as the load $L$)}. Therefore, to get a full picture of the optimal strategy we should look at how the optimal strategy depends on the state of charge.

Figure \ref{fig:3a} illustrates the optimal strategy in three cases -- where the storage is initially empty $S=0$, where the storage is half full $S=50\%$ and where the storage is full $S=100\%$. {\red As we have seen, the optimal charge/discharge strategy, for any given state of charge $S$ and level of the load $L$, $S^+_{SL}$ depends on the relationship between the current price $P_{SL}$ and the expected price in the next period $\delta E[P^+_{SL}]$. The first three graphs in figure \ref{fig:3a} illustrate $P_{SL}$ and $\delta E[P^+_{SL}]$ for each value of $L$ (here represented as the `duration'). As we saw in figure \ref{fig:3}, when the load is high (which occurs with low duration), the current price (even when the storage discharges to the minimum) is above the expected price in the next period, and the storage earns a positive payoff. When the load is low (which occurs with high duration), the current price (even when the storage charges to the maximum) is below the expected price in the next period, and the storage incurs costs. Otherwise, for a range of load in between, the current and expected future prices are the same. The storage may charge or discharge but, on average expects to not receive any gain or loss.}

As can be seen {\red (in the first three graphs of figure \ref{fig:3a})} the optimal strategy does not vary materially in this case where the storage capacity is relatively small (in contrast with figure \ref{fig:4}). Figure \ref{fig:3a} also illustrates the stationary distribution which, in this case, involves the storage facility spending almost all of its time either completely full ($S=100\%$) or completely empty ($S=0\%$).\footnote{In figures \ref{fig:3a} and \ref{fig:4}, load is assumed to be uniformly distributed on [0,100]. Raw price is a linear function of load $P(L)=1.5L+20$.}

\def\dispresults#1#2#3{
\begin{figure}[!ht]
  \begin{center}
    \begin{tikzpicture}[baseline]
      \begin{axis}[
          grid=major, 
          grid style={dashed,gray!30},
          title={Storage empty},
          xlabel=Duration (\%),
          ylabel=Price (\$/MWh),
          legend pos=north east,
          ymin=0, 
          x tick label style={rotate=90,anchor=east} 
        ]
        \addplot [mark=none, name path = A, blue]
        table[x=Load,y=PSmin,col sep=comma] {#1.csv}; 
        \addplot [mark=none, name path = B, red]
        table[x=Load,y=EPSmin,col sep=comma] {#1.csv}; 
        \addplot[gray!20] fill between[of=A and B];
      \end{axis}
    \end{tikzpicture}
    \begin{tikzpicture}[baseline]
      \begin{axis}[
          grid=major,
          grid style={dashed,gray!30},
          title={Storage half full},
          xlabel=Duration (\%),
          ylabel=Price (\$/MWh),
          legend pos=north east,
          ymin=0, 
          x tick label style={rotate=90,anchor=east} 
        ]
        \addplot [mark=none, name path=C, blue]
        table[x=Load,y=PSmid,col sep=comma] {#1.csv}; 
        \addplot [mark=none, name path=D, red]
        table[x=Load,y=EPSmid,col sep=comma] {#1.csv}; 
        \addplot[gray!20] fill between[of=C and D];
      \end{axis}
    \end{tikzpicture}
    \begin{tikzpicture}[baseline]
      \begin{axis}[
          grid=major,
          grid style={dashed,gray!30},
          title={Storage full},
          xlabel=Duration (\%),
          ylabel=Price (\$/MWh),
          legend pos=north east,
          ymin=0, 
          x tick label style={rotate=90,anchor=east} 
        ]
        \addplot [mark=none, name path=E, blue]
        table[x=Load,y=PSmax,col sep=comma] {#1.csv}; 
        \addplot [mark=none, name path=F, red]
        table[x=Load,y=EPSmax,col sep=comma] {#1.csv}; 
        \addplot[gray!20] fill between[of=E and F];
      \end{axis}
    \end{tikzpicture}
    \begin{tikzpicture}[baseline]
      \begin{axis}[
          grid=major,
          grid style={dashed,gray!30},
          title={Stationary dist.},
          xlabel=State of charge (\%),
          ylabel=Probability,
          legend pos=north east,
          ymin=0,
          x tick label style={rotate=90,anchor=east} 
        ]
        \addplot [mark=none, name path=G, black]
        table[x=s,y=xs,col sep=comma] {#1-sd.csv}; 
      \end{axis}
    \end{tikzpicture}
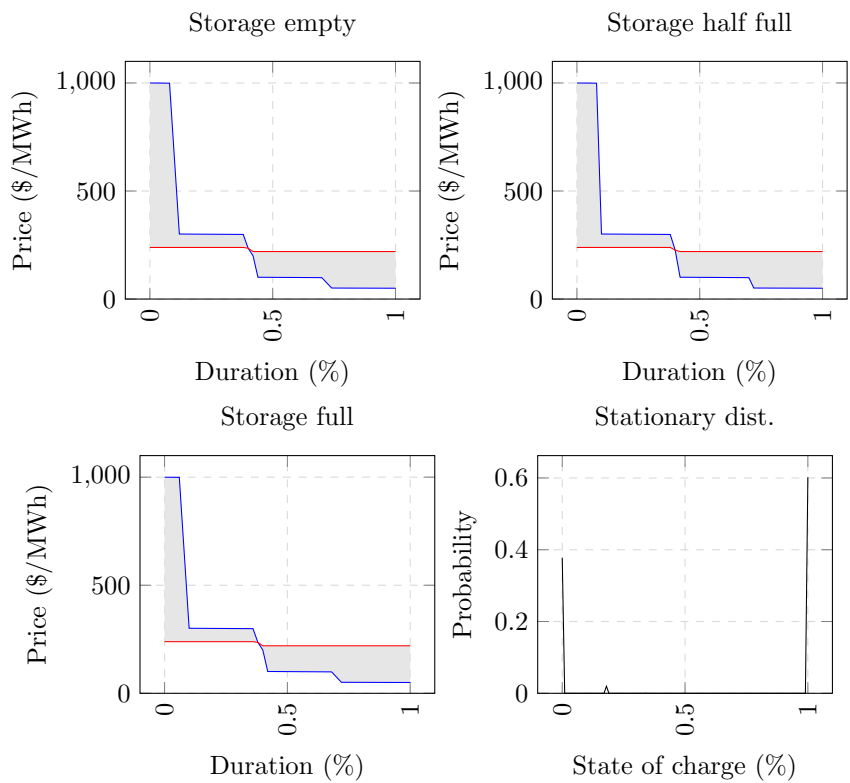
    \caption{Storage size $S=#2\%$ of load variation.}
    \label{#3}
  \end{center}
\end{figure}}

\dispresults{data/caseB-Smax10}{10}{fig:3a}

Figure \ref{fig:4} illustrates the case where the storage capacity is large relative to the variation in demand -- in this case, where the storage capacity is 150\% of the variation in demand. As can be seen, the optimal strategy varies according to the state of charge of the storage. When the storage is close to full, if the load is low {\red (represented here as high duration)}, it is efficient to charge the storage to the maximum (reflected in the grey area on the RHS of the graph on the lower left). For intermediate levels of the state of charge, the storage is not charged to the maximum or minimum, but only to the level where the current spot price is equal to the future expected spot price (reflected in the graph on the top right). When the storage is nearly empty, when the realisation of the load is high, it is efficient to discharge the storage to the minimum (grey area of the graph on the top left). The stationary distribution in this case still has substantial weight on the extremes, but also spends more time on the intermediate values of the state of charge. This can be seen in the lower right graph of figure \ref{fig:4}.\footnote{{\red
The discontinuity in the lower right graph is an artefact of the simulation.}}


\dispresults{data/CaseB-Smax150}{150}{fig:4}

\begin{theorem}
    In a power system in which consecutive draws of $L$ are independent and identically distributed, the optimal dispatch of a non-rate-limited storage system when the state of charge is $S$ and the realisation of the load is $L$, $S^*_{SL}$ is as follows:
    \begin{align}
        \forall S,L,\; S^*_{SL}=\begin{cases}
             S^{max}, & \text{if } P_{SL}<\delta E[P^+_{SL}]\\
            S^{min},  & \text{if } P_{SL}<\delta  E[P^+_{SL}]\\
            \text{in the range } [S^{min},S^{max}], & \text{if } P_{SL}=\delta E[P^+_{SL}]
        \end{cases}\label{eqn:ses}
    \end{align}
    Here $P_{SL}$ is the `spot price' (the slope of the dispatch cost) that emerges when the storage system is used optimally, and $E[P^+_{SL}]\equiv E_{L'}[P_{S^*_{SL}L'}]$ is the expected spot price in the next period.\label{thm:1}
\end{theorem}

\begin{proof}
    From equation \ref{eqn:bm2} we can write the Lagrangian as follows:
    \begin{align}
        \forall S, \; \Lagr(S) =V(S)&= E[ DC(L-(S-S^+_{SL})/\Delta|\vec{K}) + \delta V(S^+_{SL})]\nonumber\\
        &-E[\bar{\mu}_L (S^{max}-S^+_{SL})]-E[\underline{\mu}_L (S^+_{SL}-S^{min})]
    \end{align}
    The first-order condition with respect to $S^+_{SL}$ is as follows:
    \begin{equation}
        \frac{f_L}{\Delta} \frac{\partial DC}{\partial L}(L-(S-S^+_{SL})/\Delta) +\delta f_L \frac{\partial V}{\partial S}(S^+_{SL})+f_L \bar{\mu}_L -f_L \underline{\mu}_L=0\label{eqn:foc3}
    \end{equation}
    Using the result that $P_{SL}=P(L-(S-S^+_{SL})/\Delta)=\frac{\partial DC}{\partial L}$, we can write this as:
    \begin{equation}
        P_{SL} +\delta \Delta \frac{\partial V}{\partial S}(S^+_{SL})+ \bar{\mu}_L- \underline{\mu}_L=0
    \end{equation}
    Using the envelope theorem we have that:
    \begin{align}
        \frac{\partial V}{\partial S}(S)&= \frac{\partial \Lagr}{\partial S} (S)= - E\left[ \frac{\partial DC}{\partial L}(L-(S-S^+_{SL})/\Delta)\right]\nonumber\\
        &=-E[P_{SL}]
    \end{align}
{\red Therefore:
\begin{equation}
    \frac{\partial V}{\partial S}(S^+_{SL})=-E_{L'}[P_{S^+_{SL},L'}]=-E[P^+_{SL}]
\end{equation}}

Here $E[P^+_{SL}]=E_{L'}[P_{S^+_{SL}L'}]$ is the expected price in the \textit{next} interval, given that the state of charge in the current interval is $S$. It follows that the first order condition, equation \ref{eqn:foc3}, can be written:
\begin{equation}
    -P_{SL}+\delta E[P^+_{SL}] + \bar{\mu}_{SL}-\underline{\mu}_{SL}=0\label{eqn:foc4}
\end{equation}
It follows that the socially-optimal storage dispatch is $S^*_{SL}$ as follows:
    \begin{align}
        \forall S,L, \; S^*_{SL}=\begin{cases}
            S^{max}, & \text{if } P_{SL}<\delta E[P^+_{SL}]\\
            S^{min},  & \text{if } P_{SL}>\delta E[P^+_{SL}]\\
            \text{in the range }[S^{min},S^{max}],  & \text{if } P_{SL}=\delta E[P^+_{SL}]
        \end{cases}\label{eqn:sstar}
    \end{align}
\end{proof}

\subsection{Optimal investment in storage}

Now let's consider the question of the optimal mix of generation and storage in an efficiently-configured power system.

We will make the assumption that the investor in generation or storage does not know the exact state of the power system at the time when the new asset comes into service. This could arise because there is some long and uncertain delay between the time when the social planner starts considering an investment and when that investment comes online.
In the absence of information about the exact state of charge, the best the social planner can do is to estimate the long-run probability that the power system will be in any given state. This is the stationary distribution we introduced earlier.

We will assume that the social planner will choose to invest if and only if the investment is socially valuable, when averaged over the stationary distribution over states. 

Recall that solving equation \ref{eqn:bm2} yields the optimal closing state of charge $S^+$ for any given opening state of charge $S$ and realisation of load $L$. It follows that the total cost of the power system in state $S$ can be written as follows:\footnote{This is the generalisation of equation \ref{eqn:vf1}.}
\begin{align}
    V(S|\vec{K},K^S)+\frac{\sum_i f_i K_i}{1-\delta} + \frac{f_S K_S}{1-\delta}
\end{align}
Here $V(S|\vec{K},K_S)$ is as given in equation \ref{eqn:bm2}:
\begin{align}
    V(S|\vec{K},K^S) &= E[ DC(L-(S-S^+_{SL})/\Delta|\vec{K}) + \delta V(S^+_{SL}|\vec{K},K^S)]\label{eqn:bm2b}
\end{align}
It follows that the optimal capacity of generation type $i$ and the optimal capacity of storage must satisfy the following equations:
\begin{equation}
    \frac{\partial V}{\partial K_i} (S|\vec{K},K_S) + \frac{f_i}{1-\delta}=0\label{eqn:ocg}
\end{equation}
And:
\begin{equation}
    \frac{\partial V}{\partial K_S} (S|\vec{K},K_S)+\frac{f_S}{1-\delta}=0\label{eqn:ocs}
\end{equation}
The value of an additional unit of capacity of generation satisfies the following recursive equation (c.f., equation \ref{eqn:soi}):
\begin{equation}
    \frac{\partial V}{\partial K_i} (S|\vec{K},K_S) =E\left[ \frac{\partial DC}{\partial K_i} (L-(S-S^+_{SL})/\Delta|\vec{K}) + \delta \frac{\partial V}{\partial K_i} (S^+_{SL}|\vec{K},K_S)\right]\label{eqn:vaug}
\end{equation}
Similarly, the value of an additional unit of capacity of storage satisfies the following recursive equation:
\begin{align}
    \frac{\partial V}{\partial K_S} (S|\vec{K},K_S) &=E\left[ \bigl( \frac{\partial DC}{\partial L} (L-(S-S^+_{SL})/\Delta|\vec{K}) + \delta \frac{\partial V}{\partial S} (S^+_{SL}|\vec{K},K_S) \bigr) \frac{\partial S^+_{SL}}{\partial K_S} \right]\nonumber\\
    &+\delta E\left[ \frac{\partial V}{\partial K_S}(S^+_{SL}|\vec{K}, K_S)\right]\label{eqn:vauc}
\end{align}
Both of these equations have the form:
\begin{equation}
    \forall S, \; f(S)=g(S)+\delta E_L[f(S^+_{SL})]\label{eqn:fgEc}
\end{equation}
{\red Here the expectation is taken over all possible values of the load $L$.} Following the  logic in theorem \ref{thm:app} in the Appendix, the solution (at the stationary distribution) is:
\begin{equation}
    E_S[f(S)]=\frac{E_S[g(S)]}{1-\delta}\label{eqn:sde}
\end{equation}
Here $E_S[\cdot]$ means the average over states, taking the probability of each state as in the stationary distribution.

Returning to equation \ref{eqn:vaug}, {\red we see that
\begin{equation}
    \forall i,S,\; f_i(S)=\frac{\partial V}{\partial K_i} (S|\vec{K},K_S)
\end{equation}
And, using equation \ref{eqn:soi}:}
\begin{equation}
    \forall i,s,\; g_i(s)=E_L\left[ \frac{\partial DC}{\partial K_i} (L-(s-S^+_{sL})/\Delta|\vec{K})\right]=-E_L[ (P_{sL}-c_i) I(P_{sL}\ge c_i)]
\end{equation}
Substituting in to equation \ref{eqn:sde} and using equation \ref{eqn:ocg}, it follows that, assuming the investment decision is based on the long-run stationary distribution over the state of charge, it is efficient to add capacity of generation type $i$ up to the point where the expected area under the price-duration curve and above the variable cost of the generation is equal to the cost of capacity for that generation type:
\begin{align}
    (1-\delta)E_S\left[\frac{\partial V}{\partial K_i} (S|\vec{K},K_S)\right]+f_i &= -E_{SL}[ (P_{SL}-c_i) I(P_{SL}\ge c_i)]+f_i=0\label{eqn:optcapg}
\end{align}
{\red This is a very similar result to equation \ref{eqn:soi}, except that the expectation is taken over all values of the storage and the realisation of the load. In effect the `price-duration curve' is augmented by the additional uncertainty in the level of storage. The effect of increasing the level of storage on the price-duration curve is illustrated in figure \ref{fig:pdcurves}.}

Similarly, {\red returning to equation \ref{eqn:vauc}, we see that:
\begin{equation}
    f(s)=\frac{\partial V}{\partial K_S}(s|\vec{K},K_S)
\end{equation}
And:}
\begin{align}
    g(s) &=E_L\left[ \bigl( \frac{\partial DC}{\partial L} (L-(s-S^+_{sL})/\Delta|\vec{K})+ \delta \frac{\partial V}{\partial S} (S^+_{sL}|\vec{K},K_S) \bigr) \frac{\partial S^+_{sL}}{\partial K_S} \right] \nonumber\\
    &= E_L\left[ (\delta E[P^+_{sL}]-P_{sL})I(\delta E[P^+_{sL}]\ge P_{sL})\right]
\end{align}


Substituting in to equation \ref{eqn:ocs} and using equation \ref{eqn:sde}, it follows that, assuming the investment decision is based on the long-run stationary distribution over the state of charge, it is efficient to add capacity of storage up to the point where the expected area under the expected price (discounted) and above the price-duration curve is equal to the cost of adding storage capacity:
\begin{align}
    (1-\delta) E_S\left[\frac{\partial V}{\partial K_i} (S|\vec{K},K_S)\right] + f_S &= E_S[ g(S)] + f_S\nonumber\\
    &=-E_{SL}\left[ (\delta E[P^+_{SL}]-P_{SL})I(\delta E[P^+_{SL}]\ge P_{SL})\right]+f_S\nonumber\\
    &=0\label{eqn:optcaps}
\end{align}

\begin{theorem}
       In a power system in which consecutive draws of the load are independent and identically distributed and storage is not rate-limited, and assuming that the social planner does not know what the state of charge will be at the time of investment, the optimal volume of storage is where the area below the expected price in the next interval (discounted for the time value of money) and above the price-duration curve, is equal to the fixed cost of storage $f_S$.\label{thm:2}
\end{theorem}

    

{\red Holding constant the stock of generation assets, for any given volume of storage we can work out the marginal benefit of adding an additional small increment of storage (the area under the expected price in the next interval and above the price-duration curve). 
}
As we might expect, this decreases as the storage capacity increases, as illustrated in figure \ref{fig:stationaryc}. The optimal level of capacity is where the average marginal benefit of additional capacity is equal to the marginal cost of adding storage capacity. If we assume that the cost of adding storage capacity is \$5/MWh, {\red and holding constant the stock of generation assets,} the optimal level of storage capacity in this power system is where the storage is approximately 37\% of the variation of the load.

\begin{figure}[!ht]
    \centering
    \begin{tikzpicture}
\begin{axis}[width=10 cm, height=7cm,
    xlabel={Storage capacity (\% of variation in load)},
    ylabel={Av. Marginal benefit (\$/h)},
    xmin=0, xmax=160,
    ymin=0, ymax=20,
    xtick={0,20,40,60,80,100,120,140},
    legend pos=north east,
    ymajorgrids=true,xmajorgrids=true,
    grid style=dashed
]
\addplot[name path = A,
    color=blue, smooth]
    coordinates {
    (2,17.34)(10,12.47)(20,8.25)(50,3.22)(100,1.2)(150,0.60)
    };
\addplot[name path = B,
    color=red, smooth]
    coordinates {
    (10,5)(1000,5)
    };
\addplot[dashed] coordinates {(37,0)(37,5)};
    \legend{Marginal Benefit, Marginal Cost}
    \end{axis}
\end{tikzpicture}
    \caption{The optimal level of storage capacity is where the marginal benefit from adding storage is equal to the fixed cost of storage capacity. Here the marginal benefit is calculated holding constant the stock of generation assets.}
    \label{fig:stationaryc}
\end{figure}
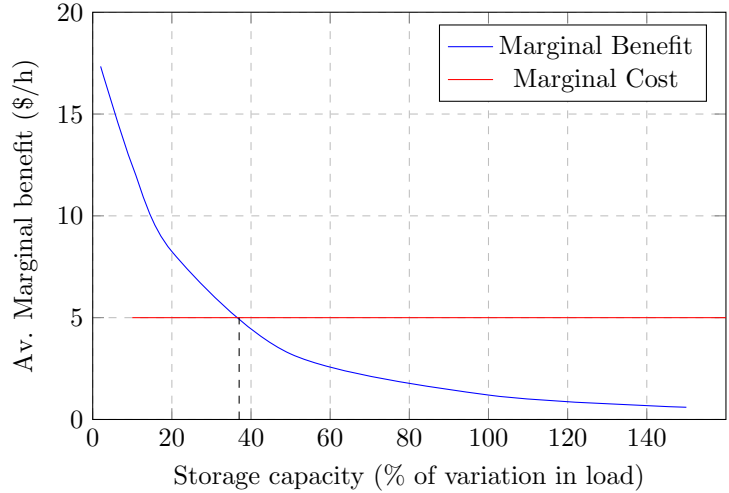

{\red In the analysis above we held constant the stock of generation assets, to find the overall optimal mix of generation and storage in a power system, we must solve equations \ref{eqn:optcapg} and \ref{eqn:optcaps} simultaneously. We can obtain some idea of the implication of changing the stock of storage on the optimal mix of generation by examining the impact of storage on the price-duration curve. It turns out that, for a given stock of generation assets, adding storage capacity `flattens' the price-duration curve around its mid-point (the average price). Figure \ref{fig:pdcurves} illustrates the price-duration curves for a range of storage capacities (storage 10\%, 50\% and 150\% of the variation in load) holding constant the underlying stock of generation assets (as reflected in $\hat{P}(L)$). As can be seen, adding storage capacity reduces the profitability of peaking generation but (since the total area under the price-duration curve remains unchanged) does not change the profitability of the baseload generation. This implies that, in the optimal mix of generation, there will be less need for peaking capacity and more need for baseload capacity.

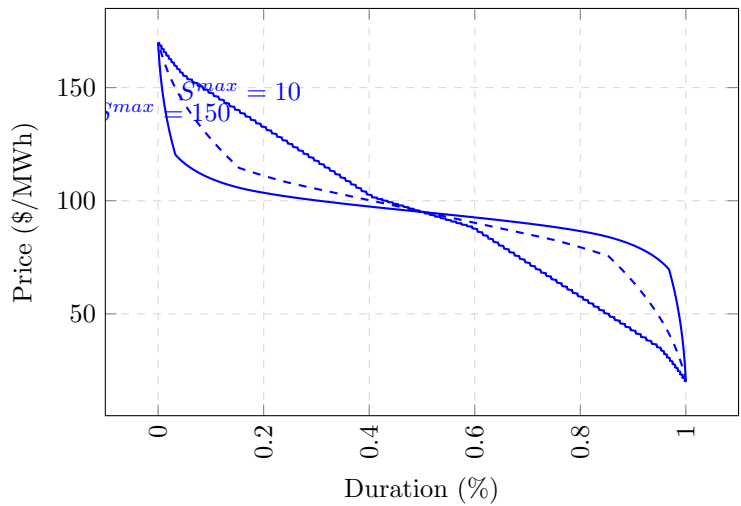
\begin{figure}[!ht]
    \centering
    \begin{tikzpicture}
\begin{axis}[width=10 cm, height=7cm,
          grid=major, 
          grid style={dashed,gray!30},
          xlabel=Duration (\%),
          ylabel=Price (\$/MWh),
          legend pos=north east,
          x tick label style={rotate=90,anchor=east} 
        ]
        \addplot [mark=none, name path = A, thick, blue]
        table[x=prob,y=price,col sep=comma] {data/CaseB-Smax10-pd.csv} 
        node [above,pos=0.8] {$S^{max}=10$};
        \addplot [mark=none, name path = A, thick, blue, dashed]
        table[x=prob,y=price,col sep=comma] {data/CaseB-Smax50-pd.csv}; 
        \addplot [mark=none, name path = A, thick, blue]
        table[x=prob,y=price,col sep=comma] {data/CaseB-Smax150-pd.csv}
        node [below,pos=0.85] {$S^{max}=150$};
    \end{axis}
\end{tikzpicture}
    \caption{Illustration of the effect of adding storage to the long-run price-duration curve. Load uniform on [0,100], Raw Price a linear function of load $P(L)=1.5L+20$. Storage 10\%, 50\% and 150\% of the variation in load.}
    \label{fig:pdcurves}
\end{figure}}

\subsection{Optimal private use of and investment in storage}\label{sec:private}

Let's now consider the decisions of a storage entrepreneur which is operating a small, price-taking storage asset. We are particularly interested to explore whether the decentralised decisions of a large number of small storage operators will result in the socially-efficient storage operation and investment decisions

We will distinguish between the state of charge of the power system as a whole and the state of charge of the small storage asset in question by using upper case $S$ to represent the state of charge of the power system as a whole and lower case $s$ to represent the state of charge of the asset in question. 

We will assume that state of charge of the power system as a whole $S$ is public information. As before, given the state of charge of the power system as a whole $S$, and the realisation of the load $L$, the state of charge of the power system as a whole is assumed to evolve according to the rule $S^+_{SL}$.

{\red Theorem \ref{thm:A} (in the Appendix) confirms that a private price-taking storage facility has an incentive to operate efficiently (that is, in accordance with the socially-optimal dispatch in theorem \ref{thm:1}). Similarly, theorem \ref{thm:B} confirms that an entrepreneur has an incentive to invest in the storage facility when it is socially-efficient to do so (in accordance with theorem \ref{thm:2}). Since the proofs of these theorems have the same structure as theorems \ref{thm:1} and \ref{thm:2}, they have been relegated to the Appendix.}

\section{Hedge contracts for storage}\label{sec:hcs}

Now let's seek an ideal hedge contract for storage.\footnote{\red As noted in footnote \ref{fn:hedge}, the case of an ideal or perfect hedge is a special case that acts as a useful benchmark or guide. In practice, if there is a cost associated with purchasing insurance, market participants are not likely to be perfectly hedged.} An ideal hedge contract is a financial contract which makes a payment $H(s^+,s, S,L)$ dependent on the {\red state of charge of the power system as a whole, $S$, the realisation of demand $L$,} and the opening and closing state of charge, $s$ and $s^+$.

As we noted in the discussion of perfect hedge contracts for generators (equations \ref{eqn:perfhc} and \ref{eqn:osch}), the perfect hedge contract must have two characteristics:
\begin{itemize}
    \item First, the hedge contract must completely eliminate the cashflow variability experienced by the storage operator when it is operating efficiently.
    
    Let $s^{+*}_{sSL}$ be the efficient strategy of {\red a small, price-taking} storage facility {\red when the state of charge of the facility is $s$ and the state of charge of all storage in the power system is $S$,} (the solution to equation \ref{eqn:bm2a}). As in equation \ref{eqn:perfhc}, {\red in order to completely eliminate the risk experienced by the facility} the perfect hedge contract must satisfy:
    \begin{equation}
        \forall s, S, L, \; \pi(s^{+*}_{sSL}|s,S,L)-H(s^{+*}_{SP},s,S,L)=\text{constant}\label{eqn:ideal1}
    \end{equation}

    \item Second, the hedge contract must not alter or distort the efficient dispatch decisions of the storage operator. In other words, the optimal dispatch for the unhedged storage facility ($s^{+*}_{sSL}$, the solution to equation \ref{eqn:bm2a}), must also be the optimal dispatch for a hedged storage facility, which is:
        \begin{align}
            V(s|S) &= \max_{s^+_{sSL}} E\left[\pi(s^{+}_{sSL}|s,S,L) -H(s^+_{sSL},s,P_{SL})+ \delta V(s^+_{sSL}|S^+_{SL})\right]\label{eqn:bm2h}\\
            & \text{subject to: }\forall s,S,L, \; s^{min} \le s^+_{sSL} \le s^{max}\nonumber
        \end{align}
\end{itemize}

Some commentators have suggested that storage systems should be hedged using a form of hedge contract known as a `collar'. A collar is a combination of a Cap and a Floor contract. A Cap contract pays out the difference between the spot price and a strike price for a pre-determined volume, but only when the spot price exceeds the strike price. Similarly, a Floor contract pays out the difference between the strike price and the spot price, but only when the spot price falls short of the strike price. Formally a cap hedge contract with a strike price of $P^S$ and a volume of $V$ makes a payment equal to $\text{Cap}(P|P^S,V)=(P-P^S)V I(P>P^S)$. Similarly, a floor hedge contract makes a payment equal to $\text{Floor}(P|P^S,V)=(P^S-P)V I(P<P^S)$.

It is straightforward to check that it is not possible to obtain a perfect hedge for a storage system using Cap and Floor contracts alone. In fact it is not possible to obtain a perfect hedge for a storage system using any hedge contract which depends only on {\red factors outside the control of the storage unit, such as the storage of the power system as a whole $S$, or the spot price $P_{SL}$. Intuitively, the reason is that the presence of the hedging eliminates the incentive on the storage facility to `invest' in charging today in order to discharge tomorrow. The hedge contract must involve financial rewards or penalties based on its own charging decisions to induce the storage facility to choose the efficient state of charge.
}

\begin{theorem}
    There is no hedge contract which {\red is independent of the storage facilities own state of charge $s^+$} which satisfies the two conditions for an ideal hedge contract above.\label{thm:4nohedge}
\end{theorem}

\begin{proof}
    {\red Let's suppose that there is such a hedge contract $H(s,S,L)$.} From equation \ref{eqn:ideal1}, the hedge contract must satisfy:
    \begin{align}
        \forall s,S,L,\; \pi(s^{+*}_{sSL}|s,S,L)-H(s,S,L)&=P_{SL} (s-s^{+*}_{sSL})-H(s,S,L)\nonumber\\
        &=\text{constant}
    \end{align}
    This implies that $H(s,S,L)=P_{SL}(s-s^{+*}_{sSL})$ up to a constant. Now consider the second condition. From equation \ref{eqn:bm2h} we must have that $s^*_{sSL}$ is the solution to:
        \begin{align}
            V(s|S) &= \max_{s^+_{sSL}} E\left[(s-s^+_{sSL})P_{SL} -H(s,S,L)+ \delta V(s^+_{sSL}|S^+_{SL})\right]\nonumber\\
            &= \max_{s^+_{sSL}} E\left[(s-s^+_{sSL})P_{SL} -P_{SL}
            (s-s^{+*}_{sSL})+ \delta V(s^+_{sSL}|S^+_{SL})\right]\nonumber\\
            &= \max_{s^{+}_{sSL}} E[(s^{+*}_{sSL}-s^+_{sSL})P_{SL} + \delta V(s^{+}_{sSL}|S^+_{SL})]\label{eqn:thm4b}
        \end{align}
    Here the maximisation is subject to the constraints in equation \ref{eqn:bm2h}.
    
    We observe that $V(s|S)$ is independent of $s$. It follows that:
    \begin{equation}
        \frac{\partial V}{\partial s}(s|S)=0
    \end{equation}
    It follows that the first order condition for equation \ref{eqn:thm4b} for $s^+_{sSL}$ is $-P_{SL}$. As this is negative, the solution to equation \ref{eqn:thm4b} is to choose the smallest possible closing state of charge $s=s^{min}$ in all states. This is not optimal, violating the second condition.    
\end{proof}

Theorem \ref{thm:4nohedge} proves that it is not possible to obtain a perfect hedge of a storage system using {\red any conventional hedge contracts}. However, theorem \ref{thm:5hedge} below shows that it is possible to obtain a perfect hedge {\red with a slightly expanded range of hedge contracts}.

\begin{theorem}
    The following hedge contract satisfies the conditions for an ideal hedge contract:
    \begin{equation}
    {\red \forall s^+,s,S,L,\;} H(s^+,s,S,L) = P_{SL}(s-s^{+*}_{sSL})-\delta E[P^+_{SL}](s^+-s^{+*}_{sSL})-c
\end{equation}
Here $c$ is a constant, {\red and $s$ and $s^+$ are the initial and final state of charge of the} storage facility {\red and $s^{+*}_{sSL}$ is the optimal strategy of the storage facility}.\label{thm:5hedge}
\end{theorem}

\begin{proof}
With this hedge contract the hedged payoff of the storage is:
    \begin{align*}
    {\red \forall s,S,L,\;} \pi(s^+|s,S,L)-H(s^+,s,S,L)= (s^{+*}_{sSL}-s^+ )(P_{SL}-\delta E[P^+_{SL}]) + c
    \end{align*}
It is clear that when the storage chooses the efficient strategy $s^+=s^{+*}_{sSL}$, {\red the hedged payoff is constant so} the first condition is satisfied. To verify the second condition, we observe that first order conditions for the hedged storage operator (equation \ref{eqn:bm2h}) are the same as for the unhedged storage operator (equation \ref{eqn:bm2}), so the profit-maximising storage supply function $s^{+*}_{sSL}$ is unchanged.
\end{proof}

We can write this hedge contract as follows:
\begin{equation}
    H(s^+,s,S,L) = (P_{SL}-\delta E[P^+_{SL}]) (s-s^{+*}_{sSL}) + \delta E[P^+_{SL}] (s-s^+)-c\label{eqn:hedge1}
\end{equation}
 

{\red Let's define two new values: The expected spot price in the next period given that the storage is full at the start of the period, and the expected spot price in the next period given that the storage is empty at the start of the period:
\begin{align}
    EP^{max}\equiv E[P^+_{SL}|S^+_{SL}=S^{max}]\text{ and }EP^{min}\equiv E[P^+_{SL}|S^+_{SL}=S^{min}]
\end{align}
}

The perfect hedge can be expressed as the sum of a `Cap' and a `Floor' contract, with a strike price of {\red $\delta EP^{min}$ and $\delta EP^{max}$, respectively, and a new form of hedge contract, which we will refer to as an `S-shaped hedge'}, as follows:
\begin{align}
    H(s^+,s,S,L) &=(P_{SL}-\delta E[P^+_{SL}]) (s-s^{+*}_{sSL}) - \delta E[P^+_{SL}] (s^+-s)-c\nonumber\\
        &=(P_{SL}-\delta E[P^+_{SL}]) (s-s^{min})I(P_{SL}> \delta E[P^+_{SL}]) \nonumber\\
        &+(\delta E[P^+_{SL}]-P_{SL})(s^{max}-s)I(P_{SL}<\delta E[P^+_{SL}])\nonumber\\
        &- \delta E[P^+_{SL}])(s^+-s)+c\nonumber\\
        &= \text{Cap}(P_{SL}|\delta EP^{min}, s-s^{min}) + \text{Floor}(P_{SL}|\delta EP^{max},s^{max}-s)\nonumber\\
        &+(s-s^+)\delta E[P^+_{SL}]+c\label{eqn:hedge}
\end{align}
{\red The shape of these components of the perfect hedge contract are illustrated in figure \ref{fig:contract_shape}. As can be seen, the combination of the Cap and Floor take the form of a `collar' hedge contract with strike prices at $EP^{min}$ and $EP^{max}$. If the realised price $P_{SL}$ is less than $EP^{max}$, or larger than $EP^{min}$ either the Cap or the Floor has a positive payoff. Note that, although the strike price for the Cap and the Floor is independent of the state of the power system, the volume of the Cap and the Floor must be constantly adjusted to reflect the current state of charge, $s$. 


\begin{figure}
    \centering
    \begin{tikzpicture}
      \begin{axis}[width=10 cm, height=7cm,
          grid=major, 
          grid style={dashed,gray!30},
          xlabel=Duration (\%),
          ylabel=Payoff (\$/h),
          legend pos=north east,
          x tick label style={rotate=90,anchor=east} 
        ]
        \addplot [mark=none, name path = A, thick, blue]
        table[x=Load,y=Cmid,col sep=comma] {data/CaseB-Smax20.csv}; 
        \addplot [mark=none, name path = B, thick, blue]
        table[x=Load,y=Fmid,col sep=comma] {data/CaseB-Smax20.csv}; 
        \addplot [mark=none, name path = C, red]
        table[x=Load,y=Rmid,col sep=comma] {data/CaseB-Smax20.csv};
        \addplot[gray!20] fill between[of=A and B];
        \addplot [dashed,thick] coordinates {(0.32,-0) (0.32,1080)};
        \addplot [dashed,thick] coordinates {(0.68,-800) (0.68,0)};
        \legend{Cap,Floor, S-shaped}
      \end{axis}
      \node at (1.5,1.5) {$P_{SL}>\delta E[P^+_{SL}]$};
      \node at (4,1) {$P_{SL}=\delta E[P^+_{SL}]$};
      \node at (6.8,0.8) {$P_{SL}<\delta E[P^+_{SL}]$};
      \node at (2.9,-0.3) {$L^{max}_S$};
      \node at (5.7,-0.3) {$L^{min}_S$};
    \end{tikzpicture}
    \caption{Illustration of the shape of the components of the perfect hedge contract in the case where the spot price is a linear function of the load, load is uniformly distributed on [0,100], the storage capacity is 20 per cent of load, and the state of charge is in the middle of the range (i.e., 10).}
    \label{fig:contract_shape}
\end{figure}
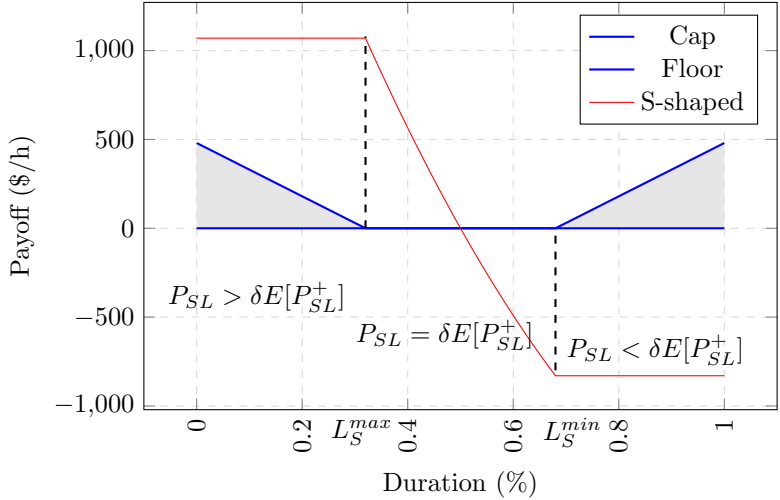

{\red The Cap and Floor contracts discussed above are reasonably standard and (although their volume varies with the state of charge of the storage system) can probably be constructed relatively easily. There will likely be more practical challenge in constructing the S-shaped hedge contract. This contract has a payoff which depends on the actual or out-turn change in the state of charge of the storage $s-s^+$ and the discounted next-period expected price $\delta E[P^+_{SL}]$. We envisage that, for any given starting value of the state of charge $S$, and given the realisation of the load $L$, the contract would specify a formula, or a look-up table for the value of $\delta E[P^+_{SL}]$.

The operation of this portfolio, for a few iterations of the power system, is illustrated in table \ref{tab:hedge}. In this example, the load is uniformly distributed on [0,100], and the storage capacity is 20. The floor contract has a strike price of \$83 and the cap has a strike price of \$107. The initial state of charge is 7.7 and the realisation of load is 32. With these values the realisation of the price (\$89.7) is equal to the discounted next-period price (\$89.7), so the storage is indifferent about charging or not. In fact, it charges to 18.7, incurring a cost of \$987.7. The payoff on the cap and floor are both zero, and the payoff on the S-shaped hedge is the expected next-period price (\$89.7) times the amount of the charge (11) which is equal to the payoff of \$987.7, so the storage is perfectly hedged.

In the third interval the state of charge is 20 and the realisation of demand is 99. The realised price (\$167) is well above the discounted expected next-period price (\$107) so the storage would like to discharge to the minimum. In doing so it earns a payoff of the spot price times the storage capacity or \$3340. The floor contract is not `in the money' but the cap contract pays out the difference between the spot price and the strike price times the storage capacity, for a payoff of \$1200. The S-shaped hedge pays out the expected price times the storage capacity, which is \$2140. Again we see that with this portfolio of hedge contracts, the storage facility is perfectly hedged.\footnote{For clarity, the payoffs in table \ref{tab:hedge} are different from the payoffs in figure \ref{fig:contract_shape} because figure \ref{fig:contract_shape} only illustrates the case where the state of charge $S=10$, whereas the state of charge in the table \ref{tab:hedge} varies each interval}}

\begin{table}[]
    \centering
    \begin{tabular}{c c c c c c c c c c c}
    \toprule
         SoC & L & P & $\delta E[P^+_{SL}]$ & $\pi$ & FV & FCF & CV & CCF & SCF & Total\\
         \midrule
         7.7 & 32 & 89.7 & 89.7 & -987.7 & 12.3 & 0 & 7.7 & 0 & -987.7 & -987.7\\
         18.7 & 29 & 87.6 & 87.6 & -112.2 & 1.3 & 0 & 18.7 & 0 & -111.2 & -111.2\\
         20.0 & 99 & 167.0 & 107.0 & 3340 & 0 & 0 & 20 & 1200 & 2140 & 3340\\ 
         \bottomrule
    \end{tabular}
    \caption{Illustration of the evolution of the payoff on the hedge contracts for three iterations of the power system, where the load is uniformly distributed on [0,100] and the storage capacity is 20. Opening state of charge is 7.7. Here L= realisation of load, P= price given SoC and load, $\pi$ is one-period payoff to storage, FV is floor volume, FCF is payoff on floor contract, CV is cap volume, CCF is payoff on cap contract, SCF is payoff on S-shaped hedge, Total is total hedge payoff (which matches storage payoff, showing the storage is perfectly hedged)}
    \label{tab:hedge}
\end{table}
}

\section{Worked example}\label{sec:examples}

Now let's consider a worked example {\red in a slightly more realistic power system}. In this example there are assumed to be three generation technologies, corresponding roughly, to baseload, mid-merit, and peaking generation. The fixed and variable cost of these generation technologies is set out in table \ref{tab:1}. The Value of Lost Load is assumed to be \$1000/MWh, which we can model as a fourth type of generation technology, with a fixed cost of zero. In addition, we will assume that storage can be added to this power system, with a fixed cost of \$25/MWh. As before the storage will be assumed to be not limited in the rate at which it can inject or withdraw power.

\begin{table}[]
    \centering
    \begin{tabular}{c|c | c}
        \toprule
         Gen & VC (\$/MWh) & FC (\$/MWh) \\
         \midrule
          L & \$50 & 185\\
          M & \$100 & 150 \\
          H & \$300 & 70 \\
          VoLL & \$1000 & \$0 \\
          \bottomrule
    \end{tabular}
    \caption{Fixed and variable costs of each generation technology in the worked example}
    \label{tab:1}
\end{table}

As can be seen in figures \ref{fig:5a}-\ref{fig:5c}, as we add more storage to the power system, the variation between the current spot price and the expected future spot price reduces (although, as the figures show, this depends on the state of charge on the storage). 


\dispresults{data/CaseA-Smax2}{2}{fig:5a}

\dispresults{data/CaseA-Smax20}{20}{fig:5b}

\dispresults{data/CaseA-Smax50}{50}{fig:5c}





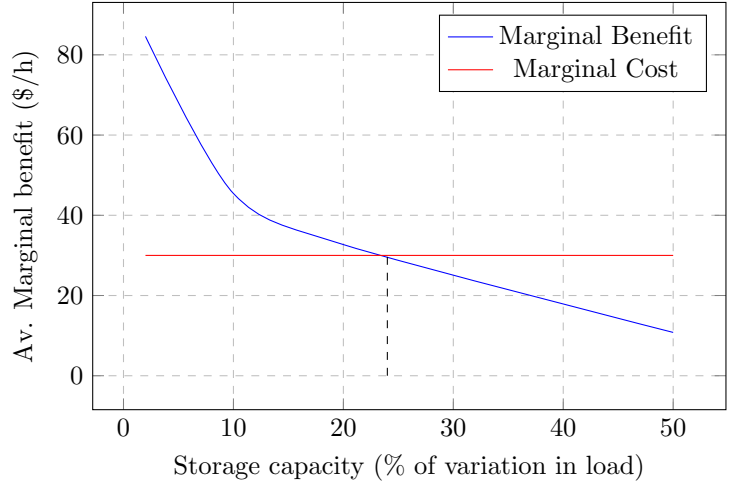
\begin{figure}[!ht]
    \centering
    \begin{tikzpicture}
\begin{axis}[width=10 cm, height=7cm,
    xlabel={Storage capacity (\% of variation in load)},
    ylabel={Av. Marginal benefit (\$/h)},
    legend pos=north east,
    ymajorgrids=true,xmajorgrids=true,
    grid style=dashed
]
\addplot[name path = A,
    color=blue, smooth]
    coordinates {
    (2,84.6)(10,45.44)(20,32.7)(50,10.8)
    };
\addplot[name path = B,
    color=red, smooth]
    coordinates {
    (2,30)(50,30)
    };
\addplot[dashed] coordinates {(24,0)(24,30)};
    \legend{Marginal Benefit, Marginal Cost}
    \end{axis}
\end{tikzpicture}
    \caption{In the worked example the optimal level of storage (holding constant the stock of generation assets) is approximately 24 per cent of the variation in the load}
    \label{fig:stationaryb}
\end{figure}

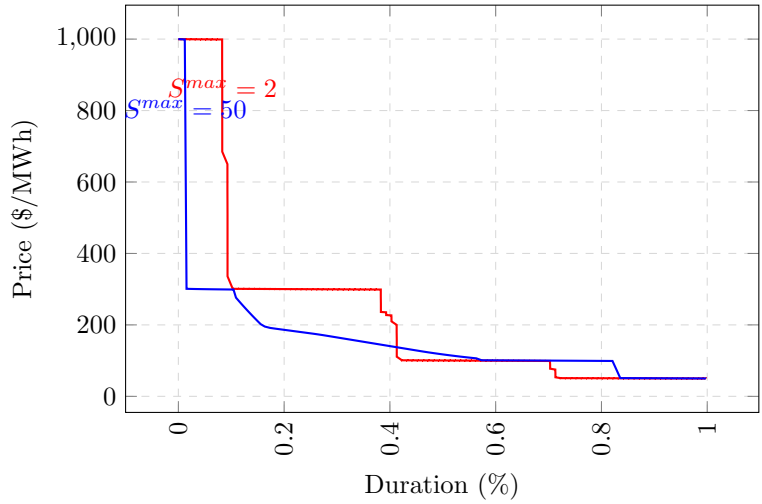
\begin{figure}[!ht]
    \centering
    \begin{tikzpicture}
\begin{axis}[width=10 cm, height=7cm,
          grid=major, 
          grid style={dashed,gray!30},
          xlabel=Duration (\%),
          ylabel=Price (\$/MWh),
          legend pos=north east,
          x tick label style={rotate=90,anchor=east} 
        ]
        \addplot [mark=none, name path = A, thick, red]
        table[x=prob,y=price,col sep=comma] {data/CaseA-Smax2-pd.csv} 
        node [above,pos=0.8] {$S^{max}=2$};
        \addplot [mark=none, name path = A, thick, blue]
        table[x=prob,y=price,col sep=comma] {data/CaseA-Smax50-pd.csv}
        node [below,pos=0.85] {$S^{max}=50$};
    \end{axis}
\end{tikzpicture}
    \caption{Illustration of the effect of adding storage on the price-duration curve in the worked example holding constant the stock of generation assets. Diagram illustrates the cases where $S^{max}=2$, and $S^{max}=50$.}
    \label{fig:pdcurves-b}
\end{figure}




\section{Conclusion}\label{sec:conclusion}

With increasing penetration of variable renewable energy it is likely that power systems of the future will require increasing volumes of energy storage services. It is therefore important to understand the economics of storage operation, investment, and hedging decisions in the context of wholesale power markets. This article contributes to that understanding by developing a theory of the optimal operation of and investment in storage systems in the context of a power system in which the realisation of load in each period is an i.i.d. random variable.

We find that many of the familiar results from the economic theory of power systems have parallels in the context of storage. But, the inclusion of storage raises new considerations and complexity in that the power system now has an additional state variable (the state of charge) which evolves over time as a Markov process.

As in the simplest power systems, it is possible the characterise the incentives for investment in storage through consideration of the price-duration curve, but now the shape and location of that price-duration curve depend on the state of charge. We also demonstrate that the private incentives for investment in a small price-taking storage facility match the social incentives. Finally we develop hedge contracts which can perfectly eliminate the risk faced by storage facilities, without distorting their operational decisions. Those hedge contracts resemble a combination of a cap and floor hedge contract, but with important differences.

{\red In this analysis we have not yet categorised the full optimal mix of storage and generation types in the power system. Doing so requires solving equations \ref{eqn:optcapg} and \ref{eqn:optcaps} simultaneously. We leave this task for future work. 

In this paper we have focused on the special case in which demand is i.i.d. In the real world demand follows a path which is partly predictable and partly uncertain from one period to the next. The impact of storage in this case is likely considerably more complex and is left for future work. In addition, this work makes important simplifying assumptions, such as that the storage is price-taking, is perfectly efficient, experiences no degradation, and is not limited in the rate at which it can charge or discharge. The impact of these assumptions can be explored in future work. 
Nevertheless,} we consider that the analysis set out here contributes to our understanding of the theory of storage in liberalised wholesale power markets.

\appendix

\section{Solutions for recursive Markov chains}

\begin{theorem}
    Suppose we have a Markov process with a state space $S$ and a transition matrix from state $s\in S$ to state $t\in S$ given by $M_{st}$. Suppose that the corresponding stationary distribution for this Markov process is given by $x_s$. By definition $x_s$ satisfies:
    \begin{equation}
        \forall t, \; x_t = \int_s x_s M_{st} ds
    \end{equation}
    Suppose that we have a recursive equation given as follows:
    \begin{equation}
        \forall s, \; f(s) = g(s) + \delta \int_t M_{st} f( t) dt \label{eqn:fsgs}
    \end{equation}
    Then we can write:
    \begin{equation}
        \int_s x_s f(s) ds = \frac{\int_s x_s g(s) ds }{1-\delta}
    \end{equation}\label{thm:app}
\end{theorem}

\begin{proof}
    Taking the expectation over equation \ref{eqn:fsgs} in the stationary distribution, we find:
    \begin{align}
        \int_s x_s f(s) ds &= \int_s x_s g(s) ds + \delta \int_s x_s \int_t M_{st} f(t) dt ds\nonumber\\
        &= \int_s x_s g(s) ds + \delta \int_t x_t f(t) dt ds
    \end{align}
    The result follows.
\end{proof}

It follows that, if we have a relationship between the state of charge of storage in the current period and the state of charge of storage in the next period as given by equation \ref{eqn:fgEc} in the text:
\begin{equation}
    f(S)=g(S)+\delta E[f(S^+_{SL})]\label{eqn:fgEc1}
\end{equation}
We can take the Markov transition matrix as follows:
\begin{equation}
    M_{st}= E[I(S^+_{sL}=t)]
\end{equation}
The solution to equation \ref{eqn:fgEc1} (at the stationary distribution) is:
\begin{equation}
    E_S[f(S)]=\frac{E_S[g(S)]}{1-\delta}
\end{equation}
Here $E_S[f(S)]$ means the expectation of the function $f(s)$ in the stationary distribution.

\section{Private incentives for use of storage}

At the start of each dispatch interval the storage system in question has some existing state of charge $s$ (MWh). This small storage asset is assumed to have state of charge limits $s^{min}\le s\le s^{max}$. The state of charge at the end of the dispatch interval will be denoted $s^+$ (MWh). Each dispatch interval, given the initial state of charge $s$, the storage system observes the realisation of the spot price $P_{SL}$ and makes a decision as to the end-of-interval state of charge $s^+$. The (one-period) cashflow of the storage system over the dispatch interval is then:
\begin{equation}
\pi(s^+|s,S,L) = P_{SL} (s-s^+) 
\end{equation}

For each opening state of charge $s$, let $V(s|S)$ be the present value of the future stream of cashflows for the small storage facility in question. We can write down the Bellman equation for $V(s|S)$. This equation is the analogue of equation \ref{eqn:osc} for storage:
\begin{align}
    V(s|S) &= \max_{s^+_{sSL}} E\left[\pi(s^+_{sSL}|s,S,L) + \delta V(s^+_{sSL}|S^+_{SL})\right]\label{eqn:bm2a}\\
    & \text{subject to: }\forall s,S,L, \; s^{min} \le s^+_{sSL} \le s^{max}\nonumber
\end{align}
Here $E[\cdot]$ is the expectation operator, taken over the realisation of $L$, and $\delta$ is the rate of time discount.

We can form the Lagrangian for equation \ref{eqn:bm2a} as follows:
\begin{align}
\Lagr(s|S) &=\sum_L \lambda_L ((s-s^+_{sSL})P_{SL} + \delta V(s^+_{sSL})) \nonumber\\
&-\sum_L \lambda_L \bar{\mu}_{sSL}(s^{max}-s^+_{sSL}) - \sum_L \lambda_L \underline{\mu}_{sSL} (s^+_{sSL} - s^{min})
\label{eqn:lagr1}
\end{align}
Here $\lambda_L>0$ is the probability distribution for $L$ and $\bar{\mu}_{sSL}$ and $\underline{\mu}_{sSL}$ are Lagrange multipliers which satisfy the the complementary slackness conditions:
\begin{equation}
    \bar{\mu}_{sSL}(s^{max}-s^+_{sSL})=0, \; \bar{\mu}_{sSL}\ge 0 \text{ and } \underline{\mu}_{sSL} (s^+_{sSL} - s^{min})=0, \; \underline{\mu}_{sSL} \ge 0
\end{equation}

From equation \ref{eqn:lagr1} the first order condition for $s^+_{sSL}$ is as follows:
\begin{equation}
    \forall s,S,L, \;\; -P_{SL}+\delta \frac{\partial V}{\partial s}(s^+_{sSL}|S^+_{SL}) + \bar{\mu}_{sSL} - \underline{\mu}_{sSL}=0\label{eqn:foc1}
\end{equation}
Using the envelope theorem we can deduce that:
\begin{equation}
    \frac{\partial V}{\partial s}(s|S)=\frac{\partial \Lagr}{\partial s}=\sum_L \lambda_L \frac{\partial \pi}{\partial s}=\sum_L \lambda_L P_{SL}=E[P_{SL}]\label{eqn:ep}
\end{equation}
Substituting back into equation \ref{eqn:foc1}, the first-order condition for $S^+$ can be written:
\begin{equation}
    -P_{SL}+\delta E_{L'}[P_{S^+_{SL}{L'}}] - \bar{\mu}_{sSL}+\underline{\mu}_{sSL}=0\label{eqn:foc2}
\end{equation}
Here $L'$ is the realisation of $L$ in the subsequent period and the expectation $E[\cdot]$ is taken over all possible values of $L'$. This proves the following theorem.

\begin{theorem}\label{thm:A}
    In a power system in which successive realisations of the spot price are independent and identically distributed, and in which the power system state of charge $S$ evolves to $S^+_{SL}$ in the next period,     the profit-maximising dispatch of a price-taking, non-rate-limited storage is to charge the storage to the maximum when the realisation of the spot price is below the average price expected in the next period (discounted by time value of money) and to discharge the storage to the minimum when the realisation of the spot price is above the average price expected in the next period (discounted by the time value of money):
    \begin{equation}
        s^*_{sSL}=\begin{cases}
            s^{max}, &\text{if } P_{SL}>\delta E[P^+_{SL}]\\
            s^{min}, &\text{if } P_{SL}<\delta E[P^+_{SL}]\\
            \text{in the range } [s^{min},s^{max}], &\text{if } P_{SL}=\delta E[P^+_{SL}]
        \end{cases}\label{eqn:pms}
\end{equation}

As the condition for the profit-maximising operation of a small storage facility (equation \ref{eqn:pms} corresponds to the condition for the welfare-maximising operation for the storage of the power system as a whole (equation \ref{eqn:ses}, we can conclude that, provided the spot price $P_{SL}$ reflects the true social marginal cost of the power system, private storage entrepreneurs will make the socially-efficient decisions regarding operation of a storage system.
\end{theorem}



Thereom \ref{thm:A} shows that a profit-maximising, price-taking storage facility which is not rate-limited behaves the same as a combination of a generator and a load with variable cost of production (for the generator, or a value of consumption for the load) equal to the expected price in the next period discounted for the time value of money $\delta E[P^+_S]=\delta E_{L'}[P_{S^+_{SL}L'}]$. However, unlike a conventional generator the capacity of this generator/load depends on the initial state of charge $s$. Specifically, a storage behaves like a generator with the rate of production $(s-s^{min})/\Delta$ when the spot price is high and behaves like a load with a rate of consumption $(s^{max}-s)/\Delta$ when the spot price is low .

Note that the storage facility does not charge or discharge (and make or receive the corresponding payments) \textit{every} dispatch interval. In the event that the spot price is above $\delta E[P^+_S]$ the facility discharges to $s^{min}$ and earns revenue for the injection. In subsequent periods, if the spot price is above $\delta E[P^+_S]$ there is no further opportunity for discharging and no cashflow. This may happen for several periods in a row.\footnote{Indeed there is no upper bound on the number of dispatch intervals in which the storage may neither charge nor discharge, but the probability of this happening for a large number of dispatch intervals in a row is extremely small.} Similarly, if the spot price is below $\delta E[P^+_S]$, the facility charges to $s^{max}$, incurring the cost of charging. In subsequent periods, if the spot price remains below $\delta E[P^+_S]$ there is no further charging and no cashflow.

It is worth noting that the present value of the stream of future cashflows of the storage facility varies both with its own state of charge $s$ and the system-wide state of charge $S$. The higher the value of the state of charge $s$, the higher the present value of the stream of cashflows (in effect, because the facility can discharge earning positive cashflow before it is required to charge again). From equation \ref{eqn:ep}, because $E[P]>0$, $V(s|S)$ is increasing in $s$.


\section{Private incentives for investment in storage}

Now let's consider the profit-maximising investment decision for a storage facility. As before we will assume that the capacity of the storage system must be chosen before the spot price is realised. We will assume that the investment decision is materially longer-lived than the dispatch interval.


Returning to equation \ref{eqn:bm2a}, we can write this as follows:
\begin{align}
    V(s|S,K_S) &= E\left[\pi(s^+_{sSL}|s,S,L) + \delta V(s^+_{sSL}|S^+_{SL},K_S)\right]\nonumber\\
    &=E\left[ P_{SL}(s-s^+_{sSL})+ \delta V(s^+_{sSL}|S^+_{SL},K_S)\right]\label{eqn:vf2}
\end{align}
Differentiating with respect to $K_S$ we find that:
\begin{equation}
    \frac{\partial V}{\partial K_S} (s|S,K_S)= -E[\bar{\mu}_{sSL}]+\delta E\left[\frac{\partial V}{\partial K_S}(s^+_{sSL}|S^+_{SL},K_S)\right]\label{eqn:pvf1}
\end{equation}
From equation \ref{eqn:foc2} we know that 
\begin{equation}
    \bar{\mu}_{sSL}=\begin{cases}
  \delta E[P_{S^+_{SL}L^+}]-P_{SL}, & \text{if }P_{SL} \le \delta E[P_{S^+_{SL}L^+}]\text{ and}\\
  0, & \text{otherwise}
    \end{cases}
\end{equation}
As this doesn't depend on the state of charge $s$ of the storage asset in question, we can drop the dependence on $s$ equation \ref{eqn:pvf1}:
\begin{equation}
    \frac{\partial V}{\partial K_S} (S)= -E[\bar{\mu}_{SL}]+\delta E\left[\frac{\partial V}{\partial K_S}(S^+_{SL})\right]\label{eqn:pvf2}
\end{equation}
Here $E[\bar{\mu}_{SL}]$ is the area between the expected price (discounted by the time value of money) and the price-duration curve when the initial state of the power system is $S$.

Now let's make the assumption that when the storage entrepreneur makes the decision to invest he/she does not know the exact state of the power system at the time when the investment comes on line. Instead, the best the entrepreneur can do is to estimate the long-run probability distribution over the range of states of charge. This is the stationary distribution we noted earlier, which we will denote $x(S)$.

As theorem \ref{thm:2} shows, it is privately optimal to add storage capacity as long as the area under the expected price (discounted by the time value of money) and above the price-duration curve -- averaged over all states in the stationary distribution -- is equal to the cost of adding storage capacity.

\begin{theorem}
    In a power system in which the state of the power system evolves according to a Markov process, investment in capacity is long-lived, and storage investors have no information about the overall state of the power system at the time of making the investment decision, it is privately profitable to augment a price-taking, non-rate-limited storage facility if and only if the area under the expected price (discounted by the time value of money) and above the price-duration curve (for a given state of charge for the power system as a whole),  averaged over all states in the stationary distribution -- is larger than the fixed cost of adding storage capacity.
    
    \label{thm:B}
\end{theorem}
The area under the (discounted) expected price and above the price-duration curve can be seen on the right hand side of figures \ref{fig:3a} and \ref{fig:4}.

\begin{proof}
Let's suppose that there is a fixed cost per unit of storage capacity equal to $f_S$ \$/hour. It follows that when the storage has a state of charge $s$, and the system-wide state of charge is $S$, the present value of the stream of future earnings is:
\begin{equation}
    V(s|S,K_S)-\frac{f_S}{1-\delta}K_S\label{eqn:smax}
\end{equation}
Here the factor $1-\delta$ is necessary to convert the constant fixed cost into a present value.

Equation \ref{eqn:pvf2} has the following form:
\begin{equation}
    f(S)=g(S)+\delta E[f(S^+_{SL})]\label{eqn:fgE}
\end{equation}
Using the logic in theorem \ref{thm:app} in the Appendix the solution in the stationary distribution is:
\begin{equation}
    E_S[f(S)]=\frac{E_S[g(S)]}{1-\delta}
\end{equation}
Here $E_S[f(S)]$ means the expectation of the function $f(s)$ in the stationary distribution.

It follows that the efficient level of investment in storage is where:
\begin{equation}
    E_S[\frac{\partial V}{\partial K_S} (S)]+\frac{f_S}{1-\delta}
    = \frac{-E_{S,L}[\bar{\mu}_{SL}]+f_S}{1-\delta}=0\label{eqn:pvf3}
\end{equation}
In other words, the efficient level of capacity is satisfies:
\begin{equation}
    E_{S,L}[\bar{\mu}_{SL}]=f_S
\end{equation}

\end{proof}

\bibliography{references}

@book{tanaka2024economics,
  title     = {Economics of Power Systems},
  author    = {Makoto Tanaka and Antonio J. Conejo and Afzal S. Siddiqui},
  year      = {2024},
  publisher = {Springer},
  address   = {Cham},
  isbn      = {9783031504763},
  doi       = {10.1007/978-3-031-50477-0}
}

@book{biggarhesamzadeh2014,
  title={The Economics of Electricity Markets},
  author={Biggar, D. R. and Hesamzadeh, M. R.},
  year={2014},
  publisher={Wiley}
}

@article{Stoft1999,
 ISSN = {01956574, 19449089},
 URL = {http://www.jstor.org/stable/41322816},
 author = {Steven Stoft},
 journal = {The Energy Journal},
 number = {1},
 pages = {1--23},
 publisher = {International Association for Energy Economics},
 title = {Financial Transmission Rights Meet Cournot: How TCCs Curb Market Power},
 volume = {20},
 year = {1999}
}

@article{biggar2022integrated,
  title={An integrated theory of dispatch and hedging in wholesale electric power markets},
  author={Biggar, Darryl R and Hesamzadeh, Mohammad Reza},
  journal={Energy Economics},
  volume={112},
  pages={106055},
  year={2022},
  publisher={Elsevier}
}

@article{sioshansi2021energy,
  title={Energy-storage modeling: State-of-the-art and future research directions},
  author={Sioshansi, Ramteen and Denholm, Paul and Arteaga, Juan and Awara, Sarah and Bhattacharjee, Shubhrajit and Botterud, Audun and Cole, Wesley and Cortes, Andres and De Queiroz, Anderson and DeCarolis, Joseph and others},
  journal={IEEE transactions on power systems},
  volume={37},
  number={2},
  pages={860--875},
  year={2021},
  publisher={IEEE}
}

@article{sheibani2018energy,
  title={Energy storage system expansion planning in power systems: a review},
  author={Sheibani, Mohammad Reza and Yousefi, Gholam Reza and Latify, Mohammad Amin and Hacopian Dolatabadi, Sarineh},
  journal={IET Renewable Power Generation},
  volume={12},
  number={11},
  pages={1203--1221},
  year={2018},
  publisher={Wiley Online Library}
}

@article{haas2017challenges,
  title={Challenges and trends of energy storage expansion planning for flexibility provision in low-carbon power systems--a review},
  author={Haas, Jannik and Cebulla, Felix and Cao, K and Nowak, Wolfgang and Palma-Behnke, Rodrigo and Rahmann, Claudia and Mancarella, Pierluigi},
  journal={Renewable and Sustainable Energy Reviews},
  volume={80},
  pages={603--619},
  year={2017},
  publisher={Elsevier}
}

@article{graves1999opportunities,
  title={Opportunities for electricity storage in deregulating markets},
  author={Graves, Frank and Jenkin, Thomas and Murphy, Dean},
  journal={The Electricity Journal},
  volume={12},
  number={8},
  pages={46--56},
  year={1999},
  publisher={Elsevier}
}

@article{miletic2020operating,
  title={Operating and investment models for energy storage systems},
  author={Mileti{\'c}, Marija and Pand{\v{z}}i{\'c}, Hrvoje and Yang, Dechang},
  journal={Energies},
  volume={13},
  number={18},
  pages={4600},
  year={2020},
  publisher={MDPI}
}

@article{van2013optimal,
  title={Optimal control of end-user energy storage},
  author={Van De Ven, Peter M and Hegde, Nidhi and Massouli{\'e}, Laurent and Salonidis, Theodoros},
  journal={IEEE Transactions on Smart Grid},
  volume={4},
  number={2},
  pages={789--797},
  year={2013},
  publisher={IEEE}
}

@article{sioshansi2013dynamic,
  title={A dynamic programming approach to estimate the capacity value of energy storage},
  author={Sioshansi, Ramteen and Madaeni, Seyed Hossein and Denholm, Paul},
  journal={IEEE Transactions on Power Systems},
  volume={29},
  number={1},
  pages={395--403},
  year={2013},
  publisher={IEEE}
}

@article{mcconnell2015estimating,
  title={Estimating the value of electricity storage in an energy-only wholesale market},
  author={McConnell, Dylan and Forcey, Tim and Sandiford, Mike},
  journal={Applied Energy},
  volume={159},
  pages={422--432},
  year={2015},
  publisher={Elsevier}
}

@article{das2015assessing,
  title={Assessing the benefits and economics of bulk energy storage technologies in the power grid},
  author={Das, Trishna and Krishnan, Venkat and McCalley, James D},
  journal={Applied Energy},
  volume={139},
  pages={104--118},
  year={2015},
  publisher={Elsevier}
}

@article{sioshansi2009estimating,
  title={Estimating the value of electricity storage in PJM: Arbitrage and some welfare effects},
  author={Sioshansi, Ramteen and Denholm, Paul and Jenkin, Thomas and Weiss, Jurgen},
  journal={Energy economics},
  volume={31},
  number={2},
  pages={269--277},
  year={2009},
  publisher={Elsevier}
}

@article{walawalkar2007economics,
  title={Economics of electric energy storage for energy arbitrage and regulation in New York},
  author={Walawalkar, Rahul and Apt, Jay and Mancini, Rick},
  journal={Energy Policy},
  volume={35},
  number={4},
  pages={2558--2568},
  year={2007},
  publisher={Elsevier}
}

@article{figueiredo2006economics,
  title={The economics of energy storage in 14 deregulated power markets},
  author={Figueiredo, F Cristina and Flynn, Peter C and Cabral, Edgar A},
  journal={Energy Studies Review},
  volume={14},
  number={2},
  pages={131},
  year={2006},
  publisher={MCMASTER INSTITUTE FOR ENERGY STUDIES}
}

@INPROCEEDINGS{hu2010optimal,
  author={Weihao Hu and Chen, Zhe and Bak-Jensen, Birgitte},
  booktitle={IEEE PES General Meeting}, 
  title={Optimal operation strategy of battery energy storage system to real-time electricity price in Denmark}, 
  year={2010},
  volume={},
  number={},
  pages={1-7},
  doi={10.1109/PES.2010.5590194}}

@article{xi2014stochastic,
  title={A stochastic dynamic programming model for co-optimization of distributed energy storage},
  author={Xi, Xiaomin and Sioshansi, Ramteen and Marano, Vincenzo},
  journal={Energy Systems},
  volume={5},
  pages={475--505},
  year={2014},
  publisher={Springer}
}

@article{shahmohammadi2018role,
  title={The role of energy storage in mitigating ramping inefficiencies caused by variable renewable generation},
  author={Shahmohammadi, Ali and Sioshansi, Ramteen and Conejo, Antonio J and Afsharnia, Saeed},
  journal={Energy Conversion and Management},
  volume={162},
  pages={307--320},
  year={2018},
  publisher={Elsevier}
}

@inproceedings{mokrian2006stochastic,
  title={A stochastic programming framework for the valuation of electricity storage},
  author={Mokrian, Pedram and Stephen, Moff and others},
  booktitle={26th USAEE/IAEE North American Conference},
  pages={24--27},
  year={2006},
  organization={Citeseer}
}

@article{lohndorf2010optimal,
  title={Optimal day-ahead trading and storage of renewable energies—an approximate dynamic programming approach},
  author={L{\"o}hndorf, Nils and Minner, Stefan},
  journal={Energy Systems},
  volume={1},
  pages={61--77},
  year={2010},
  publisher={Springer}
}

@article{lamp2022large,
  title={Large-scale battery storage, short-term market outcomes, and arbitrage},
  author={Lamp, Stefan and Samano, Mario},
  journal={Energy Economics},
  volume={107},
  pages={105786},
  year={2022},
  publisher={Elsevier}
}

@article{andres2023storing,
  title={Storing power: Market structure matters},
  author={Andr{\'e}s-Cerezo, David and Fabra, Natalia},
  journal={The RAND Journal of Economics},
  volume={54},
  number={1},
  pages={3--53},
  year={2023},
  publisher={Wiley Online Library}
}

@article{karaduman2020economics,
  title={Economics of grid-scale energy storage},
  author={Karaduman, Omer},
  journal={Job market paper},
  year={2020}
}

@article{siddiqui2019merchant,
  title={Merchant storage investment in a restructured electricity industry},
  author={Siddiqui, Afzal S and Sioshansi, Ramteen and Conejo, Antonio J},
  journal={The energy journal},
  volume={40},
  number={4},
  pages={129--164},
  year={2019},
  publisher={SAGE Publications Sage CA: Los Angeles, CA}
}

@techreport{newbery2022,
    author = {Newbery, David},
    title = {Designing an incentive-compatible efficient Renewable Electricity Support Scheme},
    institution = {University of Cambridge, Faculty of Economics} ,
    year = {2022},
    url={\url{https://doi.org/10.17863/CAM.83977}}
}

@article{balakinroger2025,
title = "Dynamic trading strategies for storage",
author = "Sergei Balakin and Guillaume Roger",
note = "Publisher Copyright: {\textcopyright} 2025",
year = "2025",
month = jul,
doi = "10.1016/j.jedc.2025.105110",
language = "English",
volume = "176",
journal = "Journal of Economic Dynamics and Control",
issn = "0165-1889",
publisher = "Elsevier",
}

\end{document}